\documentclass[11pt,a4paper]{amsart}
\usepackage{def,amsmath,amsthm,amssymb,amscd,mathrsfs,a4wide,comment,dsfont,mathrsfs,setspace,hyperref,color}
\begin{document}
\title{Scattering properties of two singularly interacting particles on the half-line}
\author{Sebastian Egger}
\address[S.Egger]{Department of Mathematics, Technion-Israel Institute of Technology 629 Amado Building, Haifa 32000, Israel}
\email{egger@tx.technion.ac.il}

\author{Joachim Kerner}
\address[J.Kerner]{Department of Mathematics and Computer Science, FernUniversit\"{a}t in Hagen, 58084 Hagen, Germany}
\email{Joachim.Kerner@fernuni-hagen.de}

\begin{abstract}
	

We analyze scattering in a system of two (distinguishable) particles moving on the half-line $\overline{\rz}_+$ under the influence of singular two-particle interactions. Most importantly, due to the spatial localization of the interactions the two-body problem is of a non-separable nature. We will discuss the presence of embedded eigenvalues and using the obtained knowledge about the kernel of the resolvent we prove a version of the limiting absorption principle. Furthermore, by an appropriate adaptation of the Lippmann-Schwinger approach we are able to construct generalized eigenfunctions which consequently allow us to establish an explicit expression for the (on-shell) scattering amplitude. An approximation of the scattering amplitude in the weak-coupling limit is also derived. 
\end{abstract}
\maketitle
\section{Introduction}
In this paper we study scattering in a system of two (distinguishable) particles moving on the real half-line $\overline{\rz}_+=[0,\infty)$ under the influence of singular and spatially localized two-particle interactions. The formal Hamiltonian of the system shall be given by
\begin{equation}\label{FormalHamiltonianI}
H=-\frac{\partial^2}{\partial x_1^2}-\frac{\partial^2}{\partial x_2^2}+v(x_1,x_2)  [\delta(x_1)+\delta(x_2)   ]\ ,
\end{equation}
$v(x_1,x_2)=v(x_2,x_1)$ being some symmetric (real-valued) interaction potential. From the Hamiltonian it is clear that the two particles are interacting only whenever at least one of the particles is situated at the origin. Furthermore, if one chooses $v:\rz^2\rightarrow \rz$ such that $\supp v \subset B_{\varepsilon}(0)$ with $B_{\varepsilon}(0) \subset \rz^2$ being the open ball of radius $\varepsilon > 0$, then the particles are interacting only whenever one particle is situated at the origin and the other is $\varepsilon$-close to it. 

The considered model originated from the theory of many-particle quantum chaos and, in particular, the theory of many-particle quantum graphs \cite{BKSingular,BKContact}. Quantum graphs, on the other hand, are (quasi) one-dimensional systems with a (potentially) complex topology. Some twenty years ago, by showing that eigenvalue correlations exhibit a behavior predicted by random matrix theory \cite{KS97}, they turned into an important model for understanding better the quantum mechanical properties of systems that are associated with chaotic classical dynamics. As a matter of fact, it is exactly the scattering of a particle in the vertices of a quantum graph which generates a chaotic dynamics. Note that scattering in a one-particle system on a quantum graph has been well-studied, see \cite{GNUSMY06,Berkolaiko:2013} and references therein. Contrary to that and owing to the fact that there are only few many-body systems which are explicitly solvable \cite{AGHH88}, the scattering properties of many-particle quantum graphs have been much less studied in the mathematical literature \cite{lobanov2008two,MP95}. The half-line represents the simplest version of a non-compact quantum graph, however,the methods developed in this paper might prove useful in the discussion of two-particle scattering on more general graphs and of more general singular two-particle interactions as presented in \cite{BKContact,BKSingular,KM16Rep}.

As outlined in \cite{KM16}, the model to be discussed is also interesting from the point of view of applications. For example, singular many-particle interactions on graphs where already considered in \cite{MP95} in order to understand their effect on the conductivity of nanoelectronic devices. In their case, the authors imagined some complex structure in the vertices of the graph leading to interactions between the particles whenever they are close to them. Regarding our model it was argued in \cite{KM16} that the Hamiltonian \eqref{FormalHamiltonianI} can be understood as describing a system of two electrons moving in a so-called composite wire which is largely normal-conductive except for a relatively small part around the origin where it is superconducting \cite{FossheimSuperconducting}. In the superconducting part, the pairing effect of superconductivity then leads to attractive two-particle interactions (Cooper pairs) .

As shown in \cite{KM16} and as explained later in more detail, the model can be reformulated as a boundary value problem for the two-dimensional Laplacian on $\rz^2_+$ with coordinate dependent Robin boundary conditions. This reformulation of the problem then enables one to use techniques and results from the theory of elliptic boundary value problems, leaving us with an at least approachable interacting many-particle system. Besides that, it is also worth mentioning that the Hamiltonian~\eqref{FormalHamiltonianI} is associated with a non-separable quantum many-body problem. As pointed out in \cite{glasser1993solvable,glasser2005solvable}, besides being only rarely discussed, non-separable quantum many-body problems have important applications regarding the foundations of quantum mechanics as well as in condensed matter physics.

The paper is organized as follows: In Section~\ref{Model} we provide a rigorous realization of \eqref{FormalHamiltonianI} as
a self-adjoint Laplacian on $\rz^2_+$ (a domain with a non-smooth but Lipschitz boundary) being subjected to variable Robin boundary conditions and we address $H^2$-regularity of the constructed operator employing methods of \cite{Helffer:2015}. In Section~\ref{prelim} we discuss the relation of the Laplacian on $\rz^2_+$ equipped with boundary conditions with the Laplacian defined on all of $\rz^2$ with a potential being singularly supported on a hypersurface  \cite{BrascheExnerKuperinSeba92,Gesztesy:2011,Behrndt:2013,Exner:2016}. In Section~\ref{Embedded} we then continue the investigation of spectral properties of the Hamiltonian~\eqref{FormalHamiltonianI} as started in \cite{KM16} and prove the absence of embedded eigenvalues in the essential spectrum whenever $\sigma$ has bounded support. This forms the counterpart of a well-known property of certain Schr\"odinger operators in full space. However, the possible eigenvalue zero requires a special attention. We are able to prove a non-existence result using properties of harmonic functions in spatial dimension two. Section~\ref{Resolvent} is then devoted to the study of the resolvent of \eqref{FormalHamiltonianI} by a suitable adaptation on various methods of \cite{Yafaev:1992,Yafaev:2010} for which we prove several integral estimates in the appendix. Finally, in Section~\ref{Scattering} we address the scattering properties of our system establishing existence and completeness of the wave operators, constructing generalized eigenfunctions and deriving an expression for the (on-shell) scattering amplitude. This allows us to establish a version of the Birman-Schwinger principle characterizing the eigenvalues but here the Birman-Schwinger operators act on the boundary of the system rather than on the complete configuration space.  We also present a novel and explicit expression for the scatting amplitude in the weak interaction limit. 

Note that asymptotic completeness of self-adjoint Laplacians on domains with smooth and compact boundary is proved in \cite{Mantile:2016n} by a Kato-Rosenblum approach involving Schatten-von Neumann estimates of suitable differences of resolvents. In \cite{Lotoreichik:2012} the corresponding Schatten-von Neumann estimates and the Kato-Rosenblum condition are discussed for the half space with a boundary potential of (possible) unbounded support but with certain regularity and decay properties. We, on the other hand, prove completeness via a suitable adaptation of an analytic Fredholm argument for Schr\"odinger operators on full space. Contrary to \cite{Lotoreichik:2012}, our boundary possesses a corner (i.e., is Lipschitz only) and we do not impose any regularity condition on the boundary potential. Furthermore, our boundary potential is also allowed to possess unbounded support, however, the decay property is more restrictive as in \cite[Lemma~3.3.~(iv)]{Lotoreichik:2012}. We also note that generalized eigenfunctions and the scattering amplitude are studied in \cite{Mantile:2016} in the case of compact and smooth hypersurfaces, however, our approach is closer to the one in \cite{bRASCHE:1992}.

Finally, we refer to section~\ref{Notation} of the appendix for some important notation used in this paper. 
\section{The model}
\label{Model}
We consider two (distinguishable) particles moving on the half-line $\overline{\rz}_+=[0,\infty)$ and whose formal Hamiltonian is given by \eqref{FormalHamiltonianI},
%
%
$v(x_1,x_2)=v(x_2,x_1)$ being some symmetric (real-valued) interaction potential. A rigorous mathematical realization of the Hamiltonian \eqref{FormalHamiltonianI} is obtained via the construction of a suitable quadratic form on $L^2(\rz^2_+)$. 

For a function $\sigma:\rz_+\rightarrow\rz$ we always identify 
\be
\label{sigma}
\sigma(y)=-v(0,y)=-v(y,0)\ ,
\ee
and denoting by $H^1(\rz^2_+)$ the Sobolev space of order one we make the following definition.
\begin{defn}
\label{defq}
For $\sigma\in L^{\infty}(\rz_+)$, the quadratic form $(q_{\sigma},H^1(\rz^2_+))$ is defined by 
\begin{equation}
\label{QuadraticForm}
q_{\sigma}[\varphi]= \int_{\rz^2_+}|\nabla \varphi|^2\ \ud \bs{x} - \int_{\partial \rz^2_+}\sigma(y)\ |\varphi_{\bv}|^2\ \ud y\ . 
\end{equation}
\end{defn}
Note that $\varphi_{\bv} \in L^2(\partial \rz^2_+)$ is the so-called trace of $\varphi \in H^1(\rz^2_+)$, $(\cdot)_{\bv}$ being the trace map according to the well-known trace theorem for Sobolev functions \cite{Dob05}.

In \cite{KM16} the following was proved.
\begin{theorem} 
If $\sigma \in L^{\infty}(\rz_+)$ then $q_{\sigma}[\cdot]$ is densely defined, closed and bounded from below. 
\end{theorem}
Hence, according to the representation theorem of quadratic forms \cite{BEH08}, there exists a unique self-adjoint operator being associated with $q_{\sigma}[\cdot]$. This operator is the Hamiltonian of our system and shall be denoted as $-\Delta_{\sigma}$ in the following. Its domain shall be denoted by $\cD(-\Delta_{\sigma})$.
\begin{remark} Note that the sesquilinear form $s_{\sigma}(\cdot,\cdot)$ associated with \eqref{QuadraticForm} is given by
	\begin{equation}
	s_{\sigma}(\psi,\varphi)= \int_{\rz^2_+}\overline{\nabla \psi}\nabla \varphi\ \ud \bs{x} - \int_{\partial \rz^2_+}\sigma(y)\overline{\psi_{\bv}} \varphi_{\bv}\ \ud y \ . 
	\end{equation}
	\end{remark}
Furthermore, a close inspection of the form~\eqref{QuadraticForm} shows that it equals the form being associated with the two-dimensional Laplacian
\begin{equation}
-\Delta=-\frac{\partial^2}{\partial x_1^2}-\frac{\partial^2}{\partial x_2^2}
\end{equation}
defined on $L^2(\rz^2_+)$ and being subjected to Robin-boundary conditions of the form
\begin{equation}
\label{ROBINBC}
\begin{split}
\frac{\partial \varphi}{\partial n}(0,y)+\sigma(y)\varphi(0,y)&=0 \ , \\
\frac{\partial \varphi}{\partial n}(y,0)+\sigma(y)\varphi(y,0)&=0 \ .
\end{split}
\end{equation}
Here $\frac{\partial}{\partial n}$ denotes the inward pointing normal derivative along $\partial \rz^2_+$.
\begin{remark}\label{RemarkNeumann}
We note that the case $\sigma \equiv 0$ corresponds to the so called Neumann-Laplacian on $\rz^2_+$ being self-adjoint on the domain $\cD_{N}:=\{\varphi \in H^2(\rz^2_+)\ :\ \frac{\partial \varphi}{\partial n}=0 \ \text{on}\ \partial \rz^2_+\}$. This operator will also be denoted by $-\Delta_0$ in the subsequent.
\end{remark}
As a first result we will establish $H^2$-regularity of $-\Delta_{\sigma}$ for a large class of boundary potentials $\sigma$. We note that, by the representation theorem of quadratic forms, one always has $\cD(-\Delta_{\sigma}) \subset H^1(\rz^2_+)$. However, without additional regularity assumptions on $\sigma$ one cannot expect to have the inclusion $\cD(-\Delta_{\sigma}) \subset H^2(\rz^2_+)$. The difficulty of establishing $H^2$-regularity is well-known in the theory of elliptic boundary value problems and was therefore studied extensively \cite{Grisvard1992,Gri85}. In general, there are two reasons why $H^2$-regularity might fail to hold: the boundary conditions could be too irregular or the boundary of the domain itself (e.g., corners). In our case, $\rz^2_+$ is a convex Lipschitz domain with a corner at $(0,0) \in \rz^2$ of angle $\pi/2$. Using the results of \cite{Grisvard1992}, however, we can establish $H^2$-regularity around the corner, assuming $\sigma$ is Lipschitz continuous. Furthermore, employing the standard difference quotient technique \cite{GilTru83,Dob05}, $H^2$-regularity can be established away from the corner leaving us with the following statement.

\begin{theorem}\label{RegularityCorner} Assume that $\sigma:[0,\infty)    \rightarrow \rz$ is Lipschitz-continuous. Then one has $H^2$-regularity, i.e., 
	\begin{equation}
	f\in \cD(-\Delta_{\sigma})      \Rightarrow  f \in H^2(\rz^2_+)\ .
	\end{equation}
\end{theorem}
\begin{proof} We first show $H^2$-regularity on any domain $\Omega_1:=(0,R) \times (0,R)$: Let $f \in \cD(-\Delta_{\sigma})$ be given and consider $\tau_R f$ where $\tau_R \in C^{\infty}_0(\rz^2)$ is a smooth and radially symmetric cutoff-function such that $(\tau_R f)(\bs{x})=1$ for $\|\bs{x}\| \leq \sqrt{2}R$ and $\tau_R(\|\bs{x}\|) \leq 1$ elsewhere. We first note that $\tau_R f \in \cD(-\Delta_{\sigma})$. Indeed, one has $-\Delta(\tau_R f)\in L^2(\rz^2_+)$ and, since the normal derivative of $\tau_R$ vanishes due to symmetry, $\tau_R f$ fulfills the boundary conditions \eqref{ROBINBC}.

	Now, set $g:=-\Delta(\tau_R f) \in L^2(\rz^2_+)$ and consider the boundary value problem
	\begin{equation}\begin{cases}
	-\Delta u = g\ , \\
	\quad \frac{\partial u}{\partial n}+\sigma (\tau_R f)=0\ ,
		\end{cases}
	\end{equation}
	on the domain $D:=\{(x_1,x_2) \in \rz^2_+ \ : \ 0 < x_1,x_2 < 2 r_{max}\}$ where $r_{max}:=\sup_{\bs{x} \in \rz^2}\{\|\bs{x}\|: |\tau_R(\|\bs{x}\|)|>0\}$. Since $\sigma (\tau_R f) \in H^{1/2}(\partial \rz^2_+)$ there exists, according to [Remark~2.4.5,~\cite{Grisvard1992}], a solution $u \in H^2(D)$ fulfilling the boundary conditions as stated. On the other hand, it is well-known that the boundary value problem
	\begin{equation}\begin{cases}
	-\Delta v=0 \ ,\\
	\quad \frac{\partial v}{\partial n}=0\ ,
	\end{cases}
	\end{equation}
	has only solutions of the form $v(x)=\text{const.}$ when considered on $D$. As a consequence, $u-\tau_R f \in H^2(D)$ which implies that $\tau_R f \in H^2(D)$. By construction of $\tau_R$ this implies $f|_{\Omega} \in H^2(\Omega_1)$.
	
	Finally, $H^2$-regularity on any domain of the form $\Omega_2:=(R,\infty) \times (0,\infty)$ or $\Omega_3:= (0,\infty) \times (R,\infty)$ with $R > 0$ can be readily established employing the difference quotient technique, see \cite{GilTru83,Dob05,BKSingular}.
\end{proof}
\section{Some preliminaries}
\label{prelim}
\subsection{An auxiliary  system}
In the next subsection we are going to study the spectral measure of the ``free'' Laplacian $-\Delta_0$ (see Remark~\ref{RemarkNeumann}). For this and our following investigations it will be convenient to define a unitary equivalent system for $-\Delta_{\sigma}$ in general and $-\Delta_{0}$ in particular. To do this we introduce the reflection operator $\mc{R}:L^2(\rz^2_+)     \rightarrow  L^2(\rz^2)$ by
\be
\label{reflection1}
(\mc{R}\psi)(x_1,x_2):=\frac{1}{2}\psi(\pm x_1,\pm x_2)\ , \quad (\pm x_1,\pm x_2) \in \rz^2_+\ ,
\ee
and we note that 
\be
\label{reflection1b}
\mc{L}^2(\rz^2):=\ran \mc{R}
\ee
is a Hilbert subspace of $L^2(\rz^2)$ since $\mc{R}$ is a continuous operator. The corresponding inner product of $\mc{L}^2(\rz^2)$ agrees with the inner product of $L^2(\rz^2)$. 

In the following we will consider $\mc{R}$ as an operator from $L^2(\rz_+^2)$ to $\mc{L}^2(\rz^2)$. With this identification we obtain 
\begin{lemma}
\label{reflection}
We have
\be
\label{unitarya}
  \|\mc{R}\psi   \|_{L^2(\rz^2)}=  \|\psi   \|_{L^2(\rz^2_+)}\ .
\ee
Moreover, the adjoint $\mc{R}^{\ast}:\mc{L}^2(\rz^2)     \rightarrow  L^2(\rz^2_+)$ of $\mc{R}$ is given by
\be
\label{adjointreflection}
\mc{R}^{\ast}\psi={2}\psi|_{\rz^2_+}\ , 
\ee
and 
\be
\label{unitary}
\mc{R}^{-1}=\mc{R}^{\ast}
\ee
\end{lemma}
holds.
\begin{proof}
We only prove \eqref{unitarya} since the case \eqref{adjointreflection} is similar. We have
\be
\label{calc1}
\ba
&  \|\mc{R}\psi   \|^2_{L^2(\rz^2)}=\frac{1}{4}[ \int_{\rz^2_+}  |\psi(x_1,x_2)   |^2\ud \bs{x}+ \int_{\rz^2_+}  |\psi(-x_1,x_2)   |^2\ud \bs{x}\\
&\hspace{0.25cm}+ \int_{\rz^2_+}  |\psi(x_1,-x_2)   |^2\ud \bs{x}+ \int_{\rz^2_+}  |\psi(-x_1,-x_2)   |^2\ud \bs{x}]=  \|\psi   \|^2_{L^2(\rz^2_+)}
\ea
\ee
which also shows that $\mc{R}$ is injective. Since $\mc{R}$ is also surjective, $\mc{R}$ is invertible and \eqref{unitary} follows immediately from \eqref{adjointreflection} and \eqref{reflection1}.
\end{proof}
%
%
%
%
%
%
%
%
%
It is natural to define
\begin{equation}
\label{pro1}
\mc{H}^1(\rz^2):=H^1(\rz^2)\cap\mc{L}^2(\rz^2),\quad\mc{H}^2(\rz^2):=H^2(\rz^2)\cap\mc{L}^2(\rz^2)\ .
\ee
%
%
%
%
%
and
\be
\label{M}
C:=C_{x_1}\cup C_{x_2}\ , 
\ee
with $C_{x_1}:=\lgk (x_1,0);\quad x_1\in\rz\rgk$ and $C_{x_2}:=\lgk (0,x_2);\quad x_2\in\rz\rgk$. In a natural way, the reflection operator $\mc{R}$ induces on $L^2(\partial \rz^2_+)$ a continuous operator $\mc{R}_{\bv}:L^2(\partial\rz^2_+)     \rightarrow  {L}^2(C)$ 
by, $\psi\in L^2(\partial \rz^2_+)$,
\be
\label{reflectionb}
(\mc{R}_{\bv}\psi)(x):=\frac{1}{2}
\begin{cases}
\psi_{\bv}(x_1,0),&\quad x=(\pm x_1,0) , \quad x_1\geq 0\ ,\\
\psi_{\bv}(0,x_2),&\quad x=(0,\pm x_2) , \quad x_2\geq 0
\end{cases}
\ee
and we finally put
\be
\label{boundary}
\mathcal{L}^2(C):=\ran\mc{R}_{\bv}\ .
\ee
Analogously to Lemma \ref{reflection} we have, considering $\mc{R}_{\bv}$ as an operator from $L^2(\partial\rz^2_+)$ to $\mathcal{L}^2(C)$, the following statement.
\begin{lemma}
\label{reflection4}
We have
\be
\label{unitary3}
\|\sqrt{2}\mc{R}_{\bv}\psi\|_{L^2(C)}=  \|\psi   \|_{L^2(\partial\rz^2_+)}\ .
\ee 
Furthermore, the adjoint $\mc{R}_{\bv}^{\ast}:\mc{L}^2(C)     \rightarrow  L^2(\partial\rz^2_+)$ is given by
\be
\label{adjointreflection3}
\mc{R}_{\bv}^{\ast}\psi=\sqrt{2}\psi|_{\partial\rz^2_+}\ , 
\ee
and
\be
\label{unitary4}
\mc{R}_{\bv}^{-1}=\sqrt{2}\mc{R}_{\bv}^{\ast}
\ee
holds.
\end{lemma}
We are now in position to formulate the unitarily equivalent system announced beforehand.
\begin{prop}
\label{laplacian}
The Laplacian $-\Delta_{\sigma}$ is unitarily equivalent to $-\tilde{\Delta}_{\tilde{\sigma}}$ defined by the quadratic form $(q_{\tilde{\sigma}},\mc{H}^1(\rz^2))$ in $\mc{L}^2(\rz^2)$, $\tilde{\sigma}:=4\mc{R}_{\bv}\sigma$, 
\be
\label{qf}
q_{\tilde{\sigma}}(\varphi):= \int_{\rz^2}  |\nabla \varphi   |^2\ud \bs{x}- \int_{C}\tilde{\sigma}  |\varphi_{\bv}   |^2\ud y\ .
\ee
\end{prop}
\begin{proof}
We show that the quadratic form $(q_{\sigma},H^1(\rz_+^2))$ is unitarily equivalent to $(q_{\tilde{\sigma}},\mc{H}^1(\rz^2))$. By Lemma \ref{reflection}, it suffices to show that 
\be
\label{quadraticform}
q_{\sigma}(\varphi)=q_{\tilde{\sigma}}(\mc{R}\varphi)\ .
\ee
Obviously, 
\be
 \int_{\rz^2}  |\nabla (\mc{R}\varphi)   |^2\ud \bs{x}= \int_{\rz^2_+}  |\nabla \varphi   |^2\ud \bs{x}
\ee
and by Lemma \ref{reflection4} we get
\be
\label{quadratiform2}
 \int_{C}\tilde{\sigma}  |(\mc{R}\varphi)_{\bv}   |^2\ud y= \int_{C}\tilde{\sigma}  |\mc{R}_{\bv}\varphi_{\bv}   |^2\ud y= \int_{\partial\rz^2_+}{\sigma}  |\varphi_{\bv}   |^2\ud y \ .
\ee
\end{proof}
\begin{cor}
The free Laplacian $-\Delta_0$ is unitarily equivalent to $(-\Delta,\mc{H}^2(\rz^2))$.
\end{cor}
\begin{remark}
\label{remark}
The system $-\tilde{\Delta}_{\tilde{\sigma}}$ is similar to systems considered in \cite{BrascheExnerKuperinSeba92,Behrndt:2013,Exner:2016}, however, there the quadratic forms of the form \eqref{qf} were studied on $L^2(\rz^d)$ with domain $H^1(\rz^d)$ rather than $\mc{L}^2(\rz^d)$ with domain $\mc{H}^1(\rz^d)$.  
\end{remark}
\subsection{The spectral measure of the free Laplacian}
We are now going to study the spectral measure of the free Laplacian $-\Delta_0$ and we set, $i=1,2,3,4$, 
\be
\label{pro10}
\tilde{\mc{R}}_i:L^2(\rz^2)     \rightarrow  L^2(\rz^2_+)
\ee
with
\be
\label{pro11}
\tilde{\mc{R}}_i(\psi):=\frac{1}{2}\mc{R}I_i\psi\ ,
\ee
and where $I_i:L^2(\rz^2)     \rightarrow  L^2(\rz^2_+)$ is defined by, $(x_1,x_2)\in\rz_+^2$,
\be
(I_i{\psi})(x_1,x_2):=
\begin{cases}
	\psi(x_1,x_2),& i=1\ ,\\
	\psi(-x_1,x_2),& i=2\ ,\\
	\psi(-x_1,-x_2),& i=3\ ,\\
	\psi(x_1,-x_2),& i=4\ .
\end{cases}
\ee
Obviously we have $2\mc{R}_1=\mc{R}$ and we put
\be
\label{pro22}
\tilde{\mc{R}}:=\tilde{\mc{R}}_1+\tilde{\mc{R}}_2+\tilde{\mc{R}}_3+\tilde{\mc{R}}_4\ ,
\ee
enabling us to specify the orthogonal projection embedding $\mcL$ into $L^2(\rz^2)$.
\begin{lemma}
	\label{pro31}
	$\tilde{\mc{R}}:L^2(\rz^2)     \rightarrow  L^2(\rz^2)$ is the orthogonal projection onto $\mc{L}^2(\rz^2)$.
\end{lemma}
\begin{proof}
	Obviously, $\ran\tilde{\mc{R}}=\mc{L}^2(\rz^2)$. Hence, we have to show that $\tilde{\mc{R}}^{\ast}=\tilde{\mc{R}}$ and $\tilde{\mc{R}}^2=\tilde{\mc{R}}$. The first identity can be proved analogously to Lemma \ref{reflection} and the second identity follows simple by observing that $\tilde{\mc{R}}$ acts as the identity on $\mc{L}^2(\rz^2)$.
\end{proof} 
Regarding the ordinary Fourier transformation $F_{0,n}:L^2(\rz^n)     \rightarrow  L^2(\rz^n)$, 
\be
\label{fo1}
(F_{0,n}\psi)(\bs{k}):=\frac{1}{(2\pi)^{\frac{n}{2}}} \int_{\rz^n}\ue^{-\ui \langle\bs{k},\bs{x}\rangle}\psi(\bs{x})\ \ud\bs{x}\ , \quad n=1,2\ ,
\ee
we make a simple observation.
\begin{lemma}
	\label{Fourier}
	$F_{0,2}$ maps $\mc{L}^2(\rz^2)$ unitarily onto $\mc{L}^2(\rz^2)$ and we have
	\be
	\label{fo21}
	F_{0,2}=\tilde{\mc{R}}F_{0,2}\tilde{\mc{R}}+(\eins-\tilde{\mc{R}}) F_{0,2}(\eins-\tilde{\mc{R}})\ . 
	\ee
	Moreover, $F_{0,1}$ maps $\mc{L}^2(C)$ unitarily onto $\mc{L}^2(C)$.
\end{lemma}
\begin{proof}
	Since $F_{0,2}$ is unitary on $L^2(\rz^2)$ we only have to check the symmetry property. An easy calculation gives
	\be
	\label{fo2}
	\ba
	\psi(-k_1,k_2)&=\frac{1}{{2\pi}} \int_{\rz^2}\ue^{-\ui (-x_1k_1+x_2k_2)}\psi(\bs{x})\ud\bs{x}=\frac{1}{{2\pi}} \int_{\rz^2}\ue^{-\ui \langle\bs{k},\bs{x}\rangle}\psi(-x_1,x_2)\ud\bs{x}\\
	&=\psi(k_1,k_2)
	\ea
	\ee 
	and similarly $\psi(-k_1,-k_2)=\psi(k_1,-k_2)=\psi(k_1,k_2)$.
	
	Relation \eqref{fo21} follows from the fact that $\tilde{\mc{R}}F_{0,2}(\eins-\tilde{\mc{R}})=(\eins-\tilde{\mc{R}})F_{0,2}\tilde{\mc{R}}=0$ which follows from Lemma~\ref{pro31} and the first part of the proof since $\tilde{\mc{R}}$ is the projection onto the symmetric subspace and hence $(\eins-\tilde{\mc{R}})$ projects onto the anti-symmetric subspace.
	
	Finally, since $F_{0,1}$ maps $L^2(\rz)$ unitarily onto $L^2(\rz)$ it readily follows that  $F_{0,1}$ maps $\mc{L}^2(C)$ unitarily onto $\mc{L}^2(C)$.
\end{proof}
For later purpose we define 
\be
\label{fo6}
\Gamma_0:\rz_+\times L^2(\rz_+^2)     \rightarrow  L^2(\mathbb{S}^1) 
\ee
by, $\psi\in L^2(\rz^2_+)$,
\be
\label{spectralmeasure1}
\Gamma_0(\lambda)(\psi)(\omega):=\frac{1}{\sqrt{2}}{(F_{0,2}\mc{R}\psi)(\sqrt{\lambda}\omega)}\ , \quad \omega \in \mathbb{S}^1\ .
\ee
Using \eqref{spectralmeasure1} we can determine the spectral measure $E_0(\cdot)$, see  \cite[p.~75]{Yafaev:2010} and \cite[Satz~8.11]{Weidmann:2000}.
\begin{lemma}
	\label{spectralmeasure10}
	The spectral measure $E_0(\cdot)$ of $-\Delta_0$ satisfies
	\be
	\label{spectralmeasure11}
	\frac{\ud}{\ud \lambda}\langle\psi,E_0(\lambda),\psi\rangle_{L^2(\rz^2_+)}=  \|\Gamma_0(\lambda,\psi)   \|^2_{L^2(\mathbb{S}^1)}\ , \quad \lambda \in \rz_+\ .
	\ee
\end{lemma}
\begin{proof}
	By Lemma \ref{Fourier}, the ordinary Fourier transformation $F_{0,2}$ is unitary on $\mc{L}^2(\rz^2)$. Thus we get (cf. \cite[p.\,75]{Yafaev:2010}), $\bs{\lambda}:=\lambda\omega$, using Lemma \ref{reflection}, $\bs{k}:=k\omega=\sqrt{\lambda}\omega$, $\lambda,k\in\rz_+$,
	\be
	\label{spectralmeasure}
	\ba
	-\langle\psi,\Delta\phi\rangle_{L^2(\rz^2_+)}&=\langle\mc{R}\psi,\Delta\mc{R}\phi\rangle_{L^2(\rz^2)}=\langle F_{0,2}\mc{R}\psi,F_{0,2}\Delta F_{0,2}^{\ast}F_{0,2}\mc{R}\psi\rangle_{L^2(\rz^2)}\\
	&= \int_{\overline{\rz}_+}k^3 \int_{\mathbb{S}^1}\overline{(F_{0,2}\mc{R}\psi)}(k\omega)(F_{0,2}\mc{R}\phi)(k\omega)\ud \omega\ud k\\
	&= \int_{\sigma(-\Delta_0)}{\lambda} \int_{\mathbb{S}^1}{\frac{1}{\sqrt{2}}}{\overline{(F_{0,2}\mc{R}\psi)}(\sqrt{\lambda}\omega)}{\frac{1}{\sqrt{2}}}{(F_{0,2}\mc{R}\phi)(\sqrt{\lambda}\omega)}\ud \omega\ud \lambda\ .
	\ea
	\ee
	We deduce the claim by using the spectral representation of a self-adjoint operator, i.e., comparing  the last line of \eqref{spectralmeasure} with \cite[Satz~8.8]{Weidmann:2000} putting there $u(t)=t$. 
\end{proof}
In the next proposition we show that the projection $E_0(I):L^2(\rz^2_+)     \rightarrow  L^2(\rz^2_+)$, $I:=[\lambda_1,\lambda_2] \subset \overline{\rz}_+$ being some bounded interval, is actually an integral operator.
\begin{lemma}
	\label{fo3}
	For a bounded interval $I=[\lambda_1,\lambda_2]\subset \overline{\rz}_+$, the projection $E_0(I)$ is an integral operator with kernel, $\bs{\lambda}:=\lambda\omega$, 
	\be
	\label{fo4}
	E_{0}(I)(\bs{x},\bs{y})= \int_{{\lambda_1}}^{{\lambda_2}} \int_{\mh{S}^1}e(\bs{x},\bs{y},\sqrt{\lambda}\omega)\ud \omega\ud \lambda\ ,
	\ee
	where
	\be
	\label{fo5}
	\ba
	&e(\bs{x},\bs{y},\lambda\omega)=\frac{1}{(2\pi)^2}( \ue^{-\ui\lambda\langle\omega,(x_1-y_1,x_2-y_2)^T\rangle}+\ue^{-\ui\lambda\langle\omega,(x_1+y_1,x_2-y_2)^T\rangle}\\
	&\hspace{0.5cm}+\ue^{-\ui\lambda\langle\omega,(x_1-y_1,x_2+y_2)^T\rangle}+\ue^{-\ui\lambda\langle\omega,(x_1+y_1,x_2+y_2)^T\rangle})\ .
	\ea
	\ee
	Moreover, 
	\be
	\label{H1}
	\ran E_{0}(I)\subset H^1(\rz_+^2)\ .
	\ee
\end{lemma}
\begin{proof}
	We first determine the corresponding projection $\tilde{E}_0(I)$ for $-\tilde{\Delta}_0$ as given by Proposition~\ref{laplacian}. This will finally prove the claim through the relation ${\mc{R}}^{\ast}\tilde{E}_0(I){\mc{R}}=E_0(I)$. 
	
	Due to Lemma \ref{Fourier} and the spectral representation of $(-\Delta,H^2(\rz^2))$ \cite[p.~76]{Yafaev:2010} we can deduce
	\be
	\label{fo8}
	\tilde{E}_{0}(I)=\tilde{\tilde{E}}_0(I)|_{\mc{L}^2(\rz^2)}\ ,
	\ee
	$\tilde{\tilde{E}}_0(I)$ being the corresponding spectral projection of $(-\Delta,H^2(\rz^2))$. From \cite[p.~19]{Weidmann:2000} we get 
	\be
	\label{H11}
	\tilde{\tilde{E}}_0(I)=F_{0,2}\kappa_IF_{0,2}^{-1}\ ,
	\ee
	where $\kappa_I$ is the characteristic function of $I$. If $f$ denotes the integral kernel of $F_{0,2}\kappa_IF_{0,2}^{-1}$ and $\mc{R}^{\ast}f\mc{R}$ the integral kernel of ${\mc{R}}^{\ast}\tilde{E}_0(I){\mc{R}}$, we obtain 
	\be
	\label{fo88}
	\ba
	& \int_{\rz_+^2}( \mc{R}^{\ast}f\mc{R})(x_1,x_2,y_1,y_2)\psi(y_1,y_2)\ud y_1\ud y_2=\mc{R}^{\ast} \int_{\rz^2} f(x_1,x_2,y_1,y_2) \mc{R}\psi(y_1,y_2)\ud y_1\ud y_2\\
	&\hspace{0.25cm}= \int_{\rz^2_+}(  f(x_1,x_2,y_1,y_2)+f(x_1,x_2,-y_1,y_2)\\
	&\hspace{0.75cm}+f(x_1,x_2,-y_1,-y_2)+f(x_1,x_2,y_1,-y_2))\psi(y_1,y_2)\ud y_1\ud y_2\ .
	\ea
	\ee
	Finally \eqref{H1} follows by \eqref{H11} using the representation 
	\be
	H^1(\rz^2):=\lgk\psi\in L^2(\rz^2);\quad \sqrt{(1+\|\cdot\|^2)}F_{0,2}\psi\in L^2(\rz^2)\rgk
	\ee 
	and \eqref{pro1}.
\end{proof}
As a direct consequence of the Lemmata \ref{spectralmeasure10} and \ref{fo3} we get a integral representation of $-\Delta_{0}$ (c.f. \cite[p.~19]{Yafaev:2010}).
\begin{prop}
	\label{prop1a}
	The identification
	\be
	\label{prop2}
	L^2(\rz^2_+)     = \int_{0}^{\infty}L^2(\sz^1)\ \ud \lambda
	\ee
	induces a map $\Gamma_0(\lambda):L^2(\rz^2_+)     \rightarrow  L^2(\sz^1)$ and $-\Delta_{0}$ acts on the r.h.s of \eqref{prop2} as a multiplication operator by $\lambda$, i.e.,
	\be
	-\Delta_{0}\psi= \int_{0}^{\infty}\lambda\Gamma_0(\lambda)\psi\ \ud \lambda\ .
	\ee
\end{prop}
The Proposition~\ref{prop1a} immediately implies by \cite[p.~18]{Weidmann:2}
\begin{cor}
	\label{ac}
	The spectrum of $-\Delta_{0}$ is purely absolutely continuous.
\end{cor}

\section{On embedded eigenvalues}\label{Embedded}
In \cite{KM16} the following was shown: For general $\sigma \in L^{\infty}(\rz_+)$ one always has $\sigma_{\ess}(-\Delta_{\sigma})=[0,\infty)$. If, in addition, $\sigma \in L^1(\rz_+)$ and 
\begin{equation}
\int_0^{\infty}\sigma(y)\ \ud y > 0\ ,
\end{equation}
then $\sigma_{d}(-\Delta_{\sigma}) \neq \emptyset$. We now address the question as to whether there exist embedded eigenvalues, i.e., eigenvalues which are contained in the essential spectrum. In general, the study of positive eigenvalues of Schr\"{o}dinger operators has a long history \cite{Kato01021967,simon1969positive}. From a physics perspective, eigenvalues at positive energies were assumed not to exist, given the potential decays sufficiently fast. However, it was eventually recognised that there indeed might exist positive eigenvalues for some potentials that decay but are highly oscillating \cite[p.~223]{ReeSim78}. 

In~\cite{KatoGrowth} Kato investigated the equation
\begin{equation}\label{KatoEquation}
-\Delta \varphi(\bs{x})-q(\bs{x})\varphi(\bs{x})=0\ , \quad \forall \bs{x} \in \rz^2:\|\bs{x}\|\geq R_0\ ,
\end{equation} 
for some $R_0 > 0$ and $q$ being some potential where $\lim_{|x|      \rightarrow  \infty}q(\bs{x}) > 0$ exists. Setting $\tilde{k}^2:=\lim_{\|\bs{x}\|      \rightarrow  \infty}q(\bs{x})$  he showed that \eqref{KatoEquation} has no solution in $L^2(\rz^2\setminus B_{R_0}(0))$ given
\begin{equation}
(2\tilde{k})^{-1}\limsup_{\|\bs{x}\|      \rightarrow  \infty}\|\bs{x}\||q(\bs{x})-\tilde{k}^2| < \frac{1}{2}\ .
\end{equation} 
Hence, by setting $q(\bs{x}):=\lambda$ for some $\lambda > 0$, we use this result to establish the following.
\begin{theorem}
\label{ee}
Assume that $\sigma \in L^{\infty}(\rz_+)$ has bounded support. Then $-\Delta_{\sigma}$ does not possess positive eigenvalues, i.e., $\sigma_{pp}(-\Delta_{\sigma})\cap (0,\infty) = \emptyset$.
\end{theorem}
\begin{proof}
Assume that $\varphi \in \cD(-\Delta_{\sigma})$ is an eigenfunction of $-\Delta_{\sigma}$ to the eigenvalue $\lambda > 0$. Pick $L > 0$ such that $\supp \sigma \subset [0,L]$ and choose a suitable (radially symmetric) cutoff-function $\tau(r) \in C^{\infty}(\rz)$ with $|\tau(r)| \leq 1$, $\tau(r)=1$ for all $r \geq 2L$ and $\tau(r)=0$ for all $r \leq L$. Then $\tau\varphi \in H^2(\rz^2_+)$ fulfills Neumann boundary conditions along the coordinate axes. Obviously, $\mc{R}(\tau\varphi)$ is a function on $\rz^2$ such that $\mc{R}(\tau\varphi) \in H^2(\rz^2)$.
	
Now, by the result of Kato we conclude that $\mc{R}(\tau\varphi)|_D=\mc{R}\varphi|_D=0$ where $D:=\{\bs{x} \in \rz^2 \ : \ \|\bs{x}\|\geq 2L \}$.

Finally, in order to show that $\varphi|_{\rz^2_+ \setminus D}=0$ we employ the following result as proved in \cite{ReeSim78}: if $\varphi=0$ in a small neighborhood $U \subset B_{r}(\bs{x})$ of $\bs{x} \in \rz^2_+$ and $|\Delta \varphi|\leq \lambda |\varphi|$ in $B_{r}(\bs{x})$ then $\varphi=0$ in $B_{r}(\bs{x})$. Hence, by constructing a suitable sequence of open balls, we conclude the statement.
\end{proof}
Theorem~\ref{ee} shows that there are no positive eigenvalues whenever the boundary potential $\sigma$ has bounded support. On the other hand, it is well-known already from the theory of Schr\"{o}dinger operators that the eigenvalue zero needs special consideration \cite{berezin1991schrodinger,daboul1994quantum}. As a matter of fact, as demonstrated in \cite[p.~198]{berezin1991schrodinger}, zero could be an eigenvalue of the operator $-\Delta+V$ even for potentials $V \in C^{\infty}_0(\rz^n)$ if $n \geq 5$. 
\begin{theorem}
\label{ee1}
Assume that $\sigma \in L^{\infty}(\rz_+)$ has bounded support. Then $\lambda=0$ is not an eigenvalue to $-\Delta_{\sigma}$.
\end{theorem}
\begin{proof} We first note that an eigenfunction to an eigenvalue zero is harmonic. Hence, let $\varphi\in D(-\Delta_{\sigma})$ be a harmonic function which is not the zero function. Without loss of generality we can also assume that $\varphi$ is real valued. 
	
As in the proof of Theorem \ref{ee} we use an analogous set $D$ and we extend $\varphi|_D$ to $\mc{R}\varphi|_D$. For simplicity, we denote this extended function again by $\varphi$. Since $\varphi|_D$ is harmonic on $D$ we have, due to \cite[Theorem~9.17]{Sheldon:2001}, the locally convergent expansion
	\be
	\label{le}
	\tilde{\varphi}(r,\theta):=\varphi(x_1(r,\theta),x_2(r,\theta))|_D=b\ln r+\sum\limits_{l=-\infty}^{\infty}(c_l r^l+\overline{c_{-l}}r^{-l})\ue^{\ui l\theta}\ ,
	\ee 
using polar coordinates. In a first step we want to take advantage of the fact that $\varphi|_D$ fulfills Neumann boundary conditions along the boundary of $\rz^2_+\cap D$. Evaluating $\partial_{\theta}\tilde{\varphi}(r,\theta)$ at $\theta=0$ and $\theta=\frac{\pi}{2}$ while requiring $\partial_{\theta}\tilde{\varphi}(r,\theta)=0$ for all $r\in (2L,\infty)$, $L$ as defined in the proof Theorem \ref{ee}, we see that $c_l\in\rz$ and $c_l=0$ for $l$ odd.

In a second step we exploit the fact that $\varphi$ belongs to $L^2(\rz^2)$. We split $\tilde{\varphi}(r,\theta):=b\ln r+\tilde{\varphi}_+(r,\theta)+\tilde{\varphi}_-(r,\theta)$, i.e., $\tilde{\varphi}_+$ involves the positive powers and $\tilde{\varphi}_-$ negative powers of $r$ in the series in \eqref{le} . By
\be
\label{neagtivepart}
\ba
|\tilde{\varphi}_-(r,\theta)|&\leq\frac{1}{r^2}\sum\limits_{l=0}^{\infty}\frac{2|c_{-l+2}|}{ |r|^l}
\ea
\ee
and observing that the series on the r.h.s. goes to zero for $r     \rightarrow \infty$ we may deduce that $\tilde{\varphi}_-(r,\theta)=\Or(  r^{-2})$ uniformly in $\theta$. Furthermore, since $\tilde{\varphi} \in L^2(D)$, it is possible to deduce that the expansion \eqref{le} necessarily reduces to 
\be
\label{le2}
\tilde{\varphi}(r,\theta)=\sum\limits_{l=1}^{\infty}\frac{2c_{-2l}}{ r^{2l}}\cos(2l\theta)\ .
\ee 

Using an analogous notation as in \eqref{le} we can expand $\varphi$ in $L^2(\rz_+\times(0,2\pi),r\ud r\ud\theta)$ by a polar Fourier expansion
\be
\label{le3}
\varphi(r,\theta)=\sum\limits_{m\in\gz} \int_{\rz_+}a_{m}(k)\Psi_m(rk,\theta)k\ud k
\ee
where $\Psi_m(rk,\theta):=\frac{1}{\sqrt{2\pi}}\ue^{\ui m\theta}\sqrt{k}J_m(rk)$, $J_m$ is the Bessel function of first kind \cite[p.~65]{Magnus:1966}, and
\begin{equation}
\label{le4}
a_m(k):= \int_{0}^{2\pi} \int_{\rz_+}\overline{\Psi_m(rk,\theta)}{\tilde{\varphi}}(r,\theta)r\ud r\ud\theta\ .
\end{equation}
We multiply both sides of \eqref{le3} and \eqref{le2} with $\psi(r,\theta):=\frac{\eta(r)}{\sqrt{2\pi}}\ue^{-\ui n\theta}$ where $\eta\in C_0^{\infty}((R,\infty))$ has compact support and $n \in 2\nz$. We obtain
\be \label{le5} \begin{split}
\langle\psi,\tilde{\varphi}\rangle_{L^2(\rz_+\times(0,2\pi),r\ud r\ud\theta)}&=\langle\sqrt{2\pi}\eta(r),c_{-{n}}{r}^{-n}\rangle_{L^2(\rz_+,r\ud  r)}\\
&= \int_{\rz_+} \int_{\rz_+}\eta(r)a_{-n}(k)\sqrt{k}J_{-n}(kr)r\ud r k\ud k\ .
\end{split}
\ee
Since $\eta$ was arbitrarily chosen in $C_0^{\infty}((R,\infty))$ we can employ the fundamental theorem of variation \cite[Satz~5.1]{Dob05} to infer 
\be
\label{le5a}
c_{{-n}}{r}^{-n}=\frac{1}{\sqrt{2\pi}} \int_{\rz_+}a_{-n}(k)\sqrt{k}J_{-n}(kr)k\ud k
\ee
for  $r>R$. Furthermore, since $c_{2l}\in \rz$, $J_{n}(kr)=J_{-n}(kr)$ \cite[p.~66]{Magnus:1966} and $\overline{a_{-n}(k)}=a_{n}(k)$ (see~\eqref{le4}) we also obtain 
\be
\label{lexxx}
c_{{-n}}{r}^{-n}=\frac{1}{\sqrt{2\pi}} \int_{\rz_+}a_{n}(k)\sqrt{k}J_{n}(kr)k\ud k\ .
\ee
Inserting \eqref{le5a} and \eqref{lexxx} into \eqref{le2} and comparing to \eqref{le3} implies that actually $\tilde{\varphi}$ in \eqref{le2} and $\varphi$ in \eqref{le3} have to agree in $L^2(\rz_+\times(0,2\pi),r\ud r\ud\theta)$. 

However, denoting by $\tilde{\varphi}_N$ the truncated series in \eqref{le2} involving the first $N$ terms w.r.t. $l$ one readily calculates
\be
\|\tilde{\varphi}_N\|^2_{L^2(\rz^2)}=4\pi\sum\limits_{l=1}^{N}|c_{-2l}|^2\int\limits_{0}^{\infty}r^{-4l+1}\ \ud r=\infty\ ,
\ee
producing a contradiction unless $c_{-2l}\equiv0$ for all $l$.

%
\end{proof}

\section{On the resolvent of $-\Delta_{\sigma}$}
\label{Resolvent}
In this section we will derive an expression for the resolvent $(-\Delta_{\sigma}-z)^{-1}$. In a first step, it is necessary to construct the resolvent in the case of vanishing boundary potential, i.e., $\sigma \equiv 0$. In this case, the resolvent is obtained from the resolvent of the (self-adjoint) operator $(-\Delta,H^2(\rz^2))$, i.e., the two-dimensional Laplacian defined on the Sobolev space $H^2(\rz^2)$. More explicitly, for $z \in \kz\setminus\rz_+$ and $k=\sqrt{z}$ let $\mathfrak{G}(k)(\textbf{x},\textbf{y})$ denote the integral kernel of $(-\Delta-z)^{-1}$ where $\textbf{x}=(x_1,x_2)$ and $\textbf{y}=(y_1,y_2)$ (see e.g. \cite[p.~78]{Yafaev:2010}). We define, for $\textbf{x},\textbf{y} \in \rz^2_+$,
\begin{equation}\label{ResolventKernelFree}\begin{split}
\mathfrak{G}^{(0)}(k)(\textbf{x},\textbf{y}):=& \mathfrak{G}(k)(\textbf{x},(y_1,y_2))+\mathfrak{G}(k)(\textbf{x},(-y_1,y_2)) \\
&+\mathfrak{G}(k)(\textbf{x},(y_1,-y_2))+\mathfrak{G}(k)(\textbf{x},(-y_1,-y_2))\ ,
\end{split}
\end{equation}
being an integral kernel of an operator acting on $L^2(\rz^2_+)$.
\begin{lemma}
\label{freeR}
For $-\Delta_0$, \eqref{ResolventKernelFree} is the integral kernel of $(-\Delta_0-z)^{-1}$ for $ z \in \kz\setminus\rz_+$.
\end{lemma}
\begin{proof}
We first show the $( -\Delta_0-z^{-1})$ leaves $\mcL$ invariant. Take $\psi\in\mcL$ and put $\phi:=( -\Delta_0-z^{-1})\psi$. Assume $\phi\notin\mcL$ then $(-\Delta_0-z)\phi\notin\mcL$ since the Laplacian obviously leaves $\mcL$ invariant. Hence, we have a contradiction and we deduce
\be
\label{res2}
(-\tilde{\Delta}_0-z)^{-1}=\tilde{\mc{R}}( -\Delta-z)^{-1}\tilde{\mc{R}}\ .
\ee
The invariance property together with Lemma \ref{pro31} now gives
\be
\label{res1}
\mc{R}^{\ast}(-\tilde{\Delta}_0-z)( -\tilde{\Delta}_0-z)^{-1}\mc{R}\psi=\psi
\ee
for all $\psi\in L(\rz^2_+)$.  An analogous argument as in \eqref{fo88} together with a combination of \eqref{res2} with \eqref{res1} leads to \eqref{ResolventKernelFree}.
\end{proof}
In the following we will take advantage of the fact that $\mathfrak{G}(k)(\textbf{x},\textbf{y})$ only depends on $\|\bs{x}-\bs{y}\|$ and that $\mathfrak{G}(k)$ can be expressed in terms of Bessel functions \cite[p.~78]{Yafaev:2010}, i.e., 
\begin{equation}
\mathfrak{G}(k)(\|\bs{x}-\bs{y}\|)=(2\pi)^{-1}K_0(-ik\|\bs{x}-\bs{y}\|)\ ,
\end{equation}
with $K_0$ being the modified Bessel function of second kind. Most importantly, $K_0(-ik|x|)$ allows for the following asymptotic expansions \cite[pp.~70,139]{Magnus:1966}:
\begin{equation}
\label{AsymptoticsI}
K_0(-ik|x|) \sim \sqrt{\frac{\pi}{-2ik|x|}}\ue^{+ik|x|}, \qquad |x|      \rightarrow  \infty\ ,
\end{equation}  
and 
\begin{equation}
\label{AsymptoticsII}
K_0(-ik|x|) \sim -\ln(|k||x|) , \qquad  |x|      \rightarrow  0\ .
\end{equation}  

\subsection{The operators ${G}_0$, $G_1$ and an expression for $(-\Delta_{\sigma}-z)^{-1}$}
We will employ methods of \cite{Yafaev:1992,Yafaev:2010} to construct the resolvent of $-\Delta_{\sigma}$. For $\sigma\in L^{\infty}(\partial\rz_+)$ we define two operators $G_0$ and $G_1$ acting from $H^1(\rz^2_+)$ to $L^2(\partial\rz^2_+)$ by, $\psi\in H^1(\rz^2_+)$,  
\be
\label{g0andg}
G_0(\psi)({x}):=\sqrt{|\sigma(x)|}\psi_{\bv}(x), \quad x\in\partial\rz^2_+\ ,
\ee
and
\be
\label{g0andga}
G_1(\psi)({x}):=-\sgn(\sigma(x))\sqrt{|\sigma(x)|}\psi_{\bv}(x), \quad x\in\partial\rz^2_+\ .
\ee
For the following lemma we refer to \cite[Definition~2,~p.~51]{Yafaev:1992}.
\begin{lemma}
\label{gg01}
The operator $-\Delta_{\sigma}$ is given by $-\Delta_0+G_1^{\ast}G_0$ where $G_0$ and $G_1$ are relatively $\sqrt{-\Delta_0}$-bounded. Furthermore, there exists a bounded operator $\Gamma(z)$ such that, for $z\in\rho(-\Delta_{\sigma})$, 
\be
\label{gg03}
(-\Delta_{\sigma}-z)^{-1}=(-\Delta_0+\eins)^{-\frac{1}{2}}\Gamma(z)(-\Delta_0+\eins)^{-\frac{1}{2}}
\ee
and
\be
\label{gg02}
\langle-\Delta_{\sigma}\phi,\psi\rangle_{L^2(\rz^2_+)}=\langle\psi,-\Delta_0\phi\rangle_{L^2(\rz^2_+)}+\langle G_1\psi,G_0\phi\rangle_{L^2(\partial\rz^2_+)}
\ee
holds for $\psi\in\mc{D}(-\Delta_0)$ and $\phi\in\mc{D}(-\Delta_{\sigma})$.
\end{lemma}
\begin{proof}
Using a suitable Sobolev trace theorem \cite[Satz~6.15]{Dob05} as well as $\sigma\in L^{\infty}(\rz_+)$ yields indeed that $G_0$ and $G_1$ are $\sqrt{-\Delta_0}$-bounded, i.e., bounded as a map between $H^1(\rz^2_+)$ and $L^2(\partial\rz_+^2)$. Moreover, since $\mc{D}(\sqrt{-\Delta_0})=H^1(\rz^2_+)=\mc{D}(\sqrt{-\Delta_{\sigma}})$ we can deduce \eqref{gg03} by \cite[p.~52]{Yafaev:1992}. The last property \eqref{gg02} follows from the sesquilinear form associated with \eqref{QuadraticForm}.
\end{proof}
We introduce the notation, $k=\sqrt{z}$, $z\in\kz\setminus\rz_+$,
\be
\label{Yavaef1}
\ba
B_{i}(k):=G_i(-\Delta_0-\overline{z})^{-1}:L^2(\rz^2_+)     \rightarrow  L^2(\partial\rz^2_+), \quad i=0,1\ ,
\ea
\ee
Moreover, by \cite[p.~52]{Yafaev:1992} and Lemma \ref{gg01} we may conclude that the operator
\be
\label{Yafaev2}
B(k):=G_0B_{1}(k)^{\ast}:L^2(\partial\rz^2_+)     \rightarrow  L^2(\partial\rz^2_+)
\ee
exists and is bounded (see also Lemma~\ref{cont}).
\begin{defn}
We denote by $\tilde{B}_i(k)$ the specific operator $B_i(k)$ obtained with the choice $\sigma\equiv 1$.
\end{defn}
Now, Lemma \ref{gg01} and \cite[Theorem~5,~p.~53]{Yafaev:1992} allow us to establish a preliminary expression for the resolvent of $-\Delta_{\sigma}$.
\begin{theorem}
\label{gg04}
The resolvent of $-\Delta_{\sigma}$ is given by, $z\in \rho(-\Delta_{\sigma})$, 
\be
\label{Rsigma}
(-\Delta_{\sigma}-z)^{-1}=(-\Delta_0-z)^{-1}-B_{1}(k)^{\ast}( \eins+B(k))^{-1}B_{0}(k)\ .
\ee
\end{theorem}
In the next subsection we will study the operators $B_{0}(k)$, $B_{i}(k)^{\ast}$ and $B(k)$ in more detail.
\subsection{The integral kernels of $B_{i}(k)$, $B_{i}^{\ast}(k)$ and $B(k)$}
In this section we show that $B_{i}(k)$, $B_{i}^{\ast}(k)$ and $B(k)$ are integral operators and elaborate on some regularity properties.
\begin{lemma}
\label{integralkernel2a}
Let $\sigma\in L^{\infty}(\rz_+)$ be given. For $z\in\kz\setminus\rz_+$ and $i\in \{0,1\}$, the operator $B_{i}(k):L^2(\rz^2_+)     \rightarrow  L^2(\partial\rz^2_+)$ is an integral operator with kernel, $x\in\rz_+$, $\bs{y}\in\rz^2_+$,
\be
\label{integralkernel3}
B_{i}(k)(x,\bs{y})=(-\sgn(\sigma(x)))^i\sqrt{|\sigma(x)|}\mf{G}^{(0)}(k)({x},\bs{y})\ .
\ee
\end{lemma}
\begin{proof}
Lemma \ref{integralkernel2a} is an easy consequence of Lemma \ref{freeR} and Definition \ref{g0andg}.
\end{proof}
\begin{lemma}
\label{integralkernel4}
Let $\sigma\in L^{\infty}(\rz_+)$ be given. For $z\in\kz\setminus\rz_+$ and $i\in \{0,1\}$, the operator $B_{i}(k)^{\ast}:L^2(\partial\rz^2_+)     \rightarrow  L^2(\rz^2_+)$ is an integral operator with kernel, $\bs{x}\in\rz^2_+$, $y\in\rz_+$,
\be
\label{integralkernel6}
B_{i}(k)^{\ast}(\bs{x},{y})=(-\sgn(\sigma(y)))^i\sqrt{|\sigma(y)|}\mf{G}^{(0)}(k)(\bs{x},{y})\ .
\ee
\end{lemma}
\begin{proof}
Lemma \ref{integralkernel4} follows easily from Lemma \ref{integralkernel2a} by observing that the operator $B_{i}(k)^{\ast}$ is adjoint to $B_{i}(k)$.
\end{proof}
For $B(k)$ to be an integral operator we need ${B_{i}}^{\ast}(k)$ to possess a regularity property.
\begin{lemma}
\label{cont}
Let $\sigma\in L^{\infty}(\rz_+)$ be given. For $z\in\kz\setminus\rz_+$ and $i\in \{0,1\}$, ${B_{i}}^{\ast}(k)$ maps $L^2(\partial \rz^2_+)$ into $H^1(\rz^2_+)$ continuously. 

Furthermore, if $\sigma\in L^{\infty}(\rz_+)$ is such that $\sigma(x)=\Or(|x|^{-1-\varepsilon})$, $\varepsilon>0$, $x     \rightarrow \infty$, then for $k\in\overline{\kz_+}$ the operator ${B_{i}}^{\ast}(k)$ maps $L^2(\partial \rz^2_+)$ into $H^1_{\loc}(\overline{\rz^2_+})$ continuously. 
\end{lemma} 
\begin{proof}
Regarding the first part of the statement, we first observe that it is enough to prove it for $\tilde{B}_{0}^{\ast}(k)$. Furthermore, for $z\in \kz\setminus\rz_+$ and $\phi\in L^2(\partial\rz^2_+)$, it will be enough to prove that $\partial_{x_1}\tilde{B}_{0}^{\ast}(k)(\phi)\in L^2(\rz^2_+)$ since the other cases are analogous. We have
\begin{equation}
\label{EquaitonBoundaryIntegralII}
\ba
&\|\partial_{x_1}\tilde{B}_{0}^{\ast}(k)(\phi)\|^2_{L^2(\rz_+^2)}=\\
&4 \int_{\rz_+^2} | \int_{\rz_+}(\partial_{x_1}\mathfrak{G}(k)(\sqrt{(x_1-y_1)^2+x^2_2})+\partial_{x_1}\mathfrak{G}(k)(\sqrt{(x_1+y_1)^2+y^2_2}))\phi(y_1,0)\ \ud y_1 \\
&\hspace{0.001cm}+ \int_{\rz_+}(\partial_{x_1}\mathfrak{G}(k)(\sqrt{x_1^2+(x_2-y_2)^2})+\partial_{x_1}\mathfrak{G}(k)(\sqrt{x_1^2+(x_2+y_2)^2}))\phi(0,y_2)\ \ud y_2  |^2\ud\bs{x}\ .
\ea
\end{equation}
Using $|a+b|^2 <2|a|^2+2|b|^2$ it suffices to estimate every of the four integral terms in \eqref{EquaitonBoundaryIntegralII} separately. 

Regarding the first one, using Young's inequality for the $x_1$-integration while setting $\tilde{\phi}(y_1):=\phi(0,y_1)$, we obtain
\be
\label{estimateee}
\ba
& \int_{\rz_+^2} | \int_{\rz_+}\partial_{x_1}\mathfrak{G}(k)(\sqrt{(x_1-y_1)^2+x^2_2})\tilde{\phi}(y_1)\ud y_1  |^2\ud\bs{x}\\
&= \int_{\rz^2_+} | \int_{\rz}\frac{x_1-y_1}{\sqrt{(x_1-y_1)^2+x^2_2}}\mathfrak{G}'(k)(\sqrt{(x_1-y_1)^2+x^2_2})\tilde{\phi}(y_1)\ud y_1  |^2\ud\bs{x}\\
&\leq \int_{\rz_+} \|\frac{(\cdot)}{\sqrt{(\cdot)^2+x^2_2}}\mathfrak{G}'(k)(\sqrt{(\cdot)^2+x^2_2}) \|^2_{L^1(\rz_+)}\ud x_2\cdot \|\tilde{\phi}\|^2_{L^2(\rz_+)}\\
&\leq  \|\mathfrak{G}(k)   \|^2_{L^2(\rz^2_+)}\cdot \|\tilde{\phi}\|^2_{L^2(\rz_+)}\ ,
\ea
\ee 
taking into account that $\mathfrak{G}'(k)$ has constant sign. Due to the exponential decay of $\mathfrak{G}(k)$ for large argument whenever $k\in\kz_+$, see \eqref{AsymptoticsI}, we only have to take care for small arguments of $\mathfrak{G}(k)$. However, the asymptotic relations~\eqref{AsymptoticsII} directly imply that $  \|\mathfrak{G}(k)   \|_{L^2(\rz^2_+)}$ is finite. The other terms in \eqref{EquaitonBoundaryIntegralII} can be treated similarly again using Young's inequality. 

Now let $k\in\overline{\kz_+} \setminus \kz_+$ and fix $\bs{x}\in\rz_+^2$. Due to \eqref{AsymptoticsI}, $\sigma(x)=\Or(|x|^{-1-\varepsilon})$ and H\"older's inequality we can infer that $B^{\ast}_0(k)\phi(\bs{x})$ exists with integral kernel \eqref{integralkernel6}. Replacing $H^1(\rz_+^2)$ by $H^1_{\loc}(\overline{\rz^2_+})$ the claim follows by repeating the argument of the first part of the proof.
\end{proof}
We are now in position to give the integral kernel of $B(k)$.
\begin{lemma}
\label{integralkernel2}
Let $\sigma\in L^{\infty}(\rz_+)$ be given. For $z\in\kz\setminus\rz_+$, the operator $B(k):L^2(\partial\rz^2_+)     \rightarrow  L^2(\partial\rz^2_+)$ is an integral operator with kernel, $x,y\in\rz_+$,
\be
\label{integralkernel11}
B(k)(x,y)=-\sqrt{|\sigma(x)|}\sgn(\sigma(y))\sqrt{|\sigma(y)|}\mf{G}^{(0)}(k)(x,y)\ .
\ee
\end{lemma}
\begin{proof}
By Lemma~\ref{cont} it follows that $B^{\ast}_{1}(k)(\phi)$ is in the domain for $G_0$. Applying the trace operator with a consecutive multiplication by $\sqrt{|\sigma(x)|}$ then proves the claim.
\end{proof}
We now provide a criterion for the existence of the operators $G_{i}E_{0}(I):L^2(\rz^2_+)     \rightarrow  L^2(\partial\rz^2_+)$ and $(G_{i}E_{0}(I))^{\ast}:L^2(\partial\rz^2_+)     \rightarrow  L^2(\rz^2_+)$, $i=0,1$, for a compact interval $I$. 
\begin{lemma}
\label{strongcont10}
Let $\sigma\in L^{\infty}(\rz_+)$ and $I=[\lambda_1,\lambda_2]$ be a bounded interval of $\overline{\rz}_+$. Then, $G_{i}E_{0}(I)$ is a bounded integral operator with kernel
\be
\label{strong1}
G_{i}E_{0}(I)(x,\bs{y})=(-\sgn(\sigma(x)))^i\sqrt{|\sigma(x)|}E_0(I)(x,\bs{y})\ .
\ee
Moreover, $(G_{i}E_{0}(I))^{\ast}$ is an integral operator as well possessing the kernel
\be
\label{strong2}
(G_{i}E_{0}(I))^{\ast}(\bs{x},y)=(-\sgn(\sigma(y)))^i\overline{E_0(I)(\bs{x},y)}\sqrt{|\sigma(y)|}\ .
\ee
\end{lemma}
\begin{proof}
First, by \eqref{H1} the operator $G_{i}E_{0}(I)$ is well-defined. By Lemma \ref{fo3} it is obvious that $G_{i}E_{0}(I)$ is an integral operator with kernel \eqref{strong1}. Moreover, $G_{i}E_{0}(I)$ is bounded and hence the adjoint kernel is given by \eqref{strong2}.
\end{proof}
For our scatting analysis it is helpful to know the explicit action of $(\Gamma_{0}(\lambda)(G_{i}E_{0}(I))^{\ast}$.   
\begin{lemma}
\label{strongcont1}
Let $\sigma\in L^{\infty}(\rz_+)$ be such that $\sigma(x)=\Or(|x|^{-1-\varepsilon})$, $x     \rightarrow \infty$, and $I=[\lambda_1,\lambda_2] \subset \overline{\rz}_+$ a bounded interval. Then, $\psi\in L^2(\partial\rz^2_+)$,
\be
\label{strongcont2}
\ba
&(\Gamma_{0}(\lambda)(G_{i}E_{0}(I))^{\ast}\psi)(\omega)=\\
& \hspace{0.5cm}
\begin{cases}
\frac{1}{\sqrt{2}\pi} \int_{C}( \mc{R}_{\bv}(-\sgn(\sigma)^i\sqrt{|\sigma|})\psi)(y)\ue^{-\ui\sqrt{\lambda}\langle\omega,y\rangle}\ud y\ , & \lambda\in (\lambda_1,\lambda_2)\ ,\\
0\ , & \lambda\notin [\lambda_1,\lambda_2]\ .
\end{cases}
\ea
\ee
\end{lemma}
\begin{proof}
We use Lemma \ref{strongcont10}: A straightforward calculation shows that the action of $(G_{i}E_{0}(I))^{\ast}$ on $\psi\in L^2(\partial\rz^2_+)$ is given by, $C=C_{x_1}\cup C_{x_2}$,  
\be
\label{strongcont3}
\ba
&(G_{i}E_{0}(I))^{\ast}\psi(\bs{x})\\
&\hspace{0.25cm}=\frac{1}{\pi^2} \int_{A_I}\ud\bs{k}\ \ue^{\ui\langle\bs{k},\bs{x}\rangle} \int_{C}( \mc{R}_{\bv}(-\sgn(\sigma)^i)\sqrt{|\sigma|)}\psi)(y)\ue^{-\ui\langle\bs{k},{y}\rangle}\ \ud{y}\ ,
\ea
\ee
where $A_I:=\lgk \bs{k}\in\rz^2;\ \|\bs{k}\|^2\in I\rgk$ and we used the symmetry of $\mc{R}_{\bv}(-\sgn(\sigma)^i)\sqrt{|\sigma|)}\psi$ on $C$. We treat only the integral over $C_{x_1}$ since the other case is analogous. Note that in \eqref{strongcont3} the integration over $C_{x_1}$ is up to a constant factor the ordinary one dimensional Fourier transformation ${F}_{0,1}$. Due to the asymptotics of $\sigma$ for large arguments we deduce that ${F}_{0,1}\sqrt{|\sigma|}\psi$ is in $H^{1+\varepsilon}(C_{x_1})$ (see e.g. \cite[p.~57]{Yafaev:2010}) and hence it possesses a continuous representative \cite[Satz~9.38]{Dob05} with respect to $k$ but obviously also for $\omega_{\bs{k}}$ where $\bs{k}=\omega_{\bs{k}}k$, $k\in\overline{\rz}_+$.

A check of the proof in \cite[Theorem~2.2.14]{Grafakos:2008} shows that in this situation we can apply the identity
\be
\label{delta}
\frac{1}{(2\pi)^2} \int_{\rz^2}\ue^{\ui\langle\bs{k}-\bs{k'},\bs{x}\rangle}\ud\bs{x}=\delta(\bs{k}-\bs{k'})
\ee
and we get using Fubini's theorem, $\bs{k'}=k'\omega_{\bs{k'}}$, 
\be
\label{delta2}
\ba
&(\Gamma_{0}((k')^2)(G_{i}E_{0}(I))^{\ast}\psi)(\omega_{\bs{k'}})\\
&\hspace{0.25cm}=\frac{1}{\sqrt{2}\pi(2\pi)^2} \int_{\rz^2}\ud\bs{x}\ \ue^{-\ui\langle\bs{k'},\bs{x}\rangle} \int_{A_I}\ud\bs{k}\ \ue^{\ui\langle\bs{k},\bs{x}\rangle} \int_{C}( \mc{R}_{\bv}(-\sgn(\sigma)^i\sqrt{|\sigma|})\psi)(y)\ue^{-\ui\langle\bs{k},y\rangle}\ud y\\
&\hspace{0.25cm}=\frac{1}{\sqrt{2}\pi(2\pi)^2} \int_{A_I}\ud\bs{k} \int_{\rz^2}\ud\bs{x}\ \ue^{\ui\langle\bs{k}-\bs{k'},\bs{x}\rangle} \int_{C}( \mc{R}_{\bv}(-\sgn(\sigma)^i\sqrt{|\sigma|})\psi)(y)\ue^{-\ui\langle\bs{k},y\rangle}\ud y\\
&\hspace{0.25cm}=
\begin{cases}
\frac{1}{\sqrt{2}\pi} \int_{C}( \mc{R}_{\bv}(-\sgn(\sigma)^i\sqrt{|\sigma|})\psi)(y)\ue^{-\ui\langle\bs{k'},y\rangle}\ud y,& k^2\in(\lambda_1,\lambda_2),\\
0,&\lambda\notin (\lambda_1,\lambda_2) \ .
\end{cases}
\ea
\ee
Now identifying $k'=\sqrt{\lambda}$ and $\omega_{\bs{k'}}=\omega_{\bs{\lambda}}$ proves the claim.
\end{proof}
We note that, for $\lambda\in\lgk\lambda_1,\lambda_2\rgk$, we would have to incorporate an extra factor of $1/2$ in \eqref{strongcont2} since $\bs{k}$ would be on the boundary of $A_I$, see \cite[pp.~208,209]{Folland:1992}. However, in the following these points can be neglected for having (spectral) measure zero.

According to \cite[Definition~5.6,~p.~31]{Yafaev:2010} the operators $G_{i}$, $i=0,1$, are called \textit{strongly-$\Delta_0$ smooth} iff for $\psi\in L^2(\partial\rz^2_+)$ we have,
\be
\label{hsmooth}
  \|\Gamma_{0}(\lambda)(G_{i}E_{0}(I))^{\ast}\psi   \|_{L^2(\sz^1)}<C  \|\psi   \|_{L^2(\partial\rz^2_+)}
\ee
and
\be
\label{hsmoothb}
  \|(\Gamma_{0}(\lambda_1)(G_{i}E_{0}(I))^{\ast}-\Gamma_{0}(\lambda_2)(G_{i}E_{0}(I))^{\ast})\psi   \|_{L^2(\sz^2)}<C|\lambda_1-\lambda_2|^{\theta}  \|\psi   \|^2_{L^2(\partial\rz^2_+)}
\ee
for all bounded intervals $I$ such that $\lambda,\mu$ are in the interior of $I$ and some $\theta>0$.
\begin{lemma}
\label{hsmooth1}
Under the assumptions of Lemma \ref{strongcont1}, the operators $G_{i}$, $i=0,1$, are strongly-$\Delta_0$ smooth with $\theta=\frac{\varepsilon}{2}$. Furthermore, the constant $C$ in \eqref{hsmooth} and \eqref{hsmoothb} depends on $\varepsilon$ only.
\end{lemma}
\begin{proof}
We prove \eqref{hsmoothb} and note that \eqref{hsmooth} can be proved analogously. For convenience we set
\be
f(y):=\mc{R}_{\bv}(-\sgn(\sigma(y))^i\sqrt{|\sigma(y)|})\psi(y)\ .
\ee 
The restriction of $f(y)$ on $C_{x_i}$, $i=1,2$, is of $\Or(  x_i^{-\frac{1+\varepsilon}{2}})$ and using the same method as in
\cite[Proposition~1.2,~p.~72]{Yafaev:2010}, $y:=\omega_{{y}}\|y\|$, we obtain
\be
\label{hsmooth4}
\ba
&  |( \Gamma_{0}(\lambda_1)(G_{i}E_{0}(I))^{\ast}-\Gamma(\lambda_2)(G_{i}E_{0}(I))^{\ast})\psi(\omega)   |^2\\
&\hspace{0.25cm} \leq(\frac{1}{\sqrt{2}\pi} \int_{C}  |f(y)( \ue^{-\ui\lambda_1\langle\omega,y\rangle}-\ue^{-\ui\lambda_2\langle\omega,y\rangle})   |\ud y)^2\\
&\hspace{0.25cm} \leq \tilde{C} \int_{C}(1+\|y\|^2)^{+\frac{1+\varepsilon}{2}}  |f(y)   |^2\ud y \int_{C}  |\ue^{-\ui\lambda_1\langle\omega,y\rangle}-\ue^{-\ui\lambda_2\langle\omega,y\rangle}   |^2(1+\|y\|^2)^{-\frac{1+\varepsilon}{2}}\ud y\\
&\hspace{0.25cm} \leq \tilde{\tilde{C}} \int_{C}  |\psi(y)   |^2\ud y \int_{C}\sin^2( \frac{( \lambda_1-\lambda_2)\langle\omega,\omega_{y}\rangle y}{2})(1+\|y\|^2)^{-\frac{1+\varepsilon}{2}}\ud y\\
&\hspace{0.25cm}<C|\lambda_1-\lambda_2|^{\varepsilon}\cdot   \|\psi   \|^2_{L^2(\partial\rz^2_+)}\ ,
\ea
\ee
with $C$ depending on $\varepsilon$ only. Now an integration w.r.t. $\omega$ proves the claim.
\end{proof} 
In  the next step we generalize \cite[Lemma~3.1]{BrascheExnerKuperinSeba92} extending the result to $\sigma$ with non-compact support and a suitable decay behavior. However, we only need to consult the case $k\in\kz_+$. For this we need an auxiliary lemma and we define 
\be
\label{aux2}
I_{1,n}=\lgk(x,0);\quad x\in  [n,n+1   ]\rgk,\quad I_{2,n}=\lgk(0,y);\quad y\in  [n,n+1   ]\rgk \ .
\ee
\begin{lemma}
\label{aux1}
Let $\sigma\in L^{\infty}(\rz_+)$ be such that $\sigma(x)=\Or(|x|^{-1-\varepsilon})$, $x     \rightarrow \infty$, for some $\varepsilon>0$. Then, $i=1,2$,
\be
\label{aux3}
  \|B(k)\psi   \|^2_{L^2(I_{i,n})}\leq \frac{C}{n^{1+\varepsilon}}  \|\psi   \|^2_{L^2(\partial\rz^2_+)}\ ,
\ee
with some $C>0$ being independent of $n$.
\end{lemma} 
\begin{proof}
We consider the case $i=1$ only, the case $i=2$ is analogous. A short calculation yields
\be
\label{aux5}
\ba
&\|B(k)\psi\|^2_{L^2(I_1,n)}\\
&\quad\leq C\int_{I_{1,n}} |\int_{\rz_+}\mf{G}(k)(|x-y|)\sgn(\sigma(y))\sqrt{|\sigma(y)|}\psi(y,0)\ \ud y |^2 \sqrt{|\sigma(x)|} \ud x\\
&\quad +C\int_{I_{1,n}} |\int_{\rz_+}\mf{G}(k)(|x+y|)\sgn(\sigma(y))\sqrt{|\sigma(y)|}\psi(y,0)\ \ud y |^2 \sqrt{|\sigma(x)|} \ud x \\
&\quad +C\int_{I_{1,n}} |\int_{\rz_+}\mf{G}(k)(\sqrt{x^2+y^2})\sgn(\sigma(y))\sqrt{|\sigma(y)|}\psi(0,y)\ \ud y |^2\sqrt{|\sigma(x)|} \ud x \ ,
\ea
\ee
for some constant $C > 0$. The second and third integral in \eqref{aux5} is of order $\Or(\ue^{-kn})$. For the first integral we use Lemma \ref{app1} with $\alpha=1/2+\varepsilon/2$ to obtain
\be
\label{aux6}
\|B(k)\psi\|^2_{L^2(I_{1,n})}\leq\frac{C}{n^{1+\epsilon}}  \|\psi   \|^2_{L^2(\partial\rz^2_+)} \ ,
\ee
with some $C>0$ independent of $n$ since $x\geq n$.
\end{proof}
\begin{lemma} 
\label{compactop}
Let $z\in\kz\setminus\rz_+$ and assume that $\sigma\in L^{\infty}(\rz_+)$ is such that $\sigma(x)=\Or(|x|^{-1-\varepsilon})$, $x     \rightarrow \infty$, for some $\varepsilon>0$. Then the operator $B(k)$ is compact.
\end{lemma}
\begin{proof}
In a first step one defines the operator $B^{(n)}(k)$ with an integral kernel as in \eqref{integralkernel11}, replacing $\mathfrak{G}^{(0)}(k)(\textbf{x}-\textbf{y})$ by
\begin{equation}\label{EquationXXX}
\mathfrak{G}^{(n)}(k)(\textbf{x}-\textbf{y})=
\begin{cases}
\mathfrak{G}^{(0)}(k)(\textbf{x}-\textbf{y}) \quad \mbox{if} \quad   \|\textbf{x}-\textbf{y}\| > \frac{1}{n}\ , \\
0 \quad \mbox{otherwise}\ .
\end{cases}
\end{equation}
One then shows that $B^{(n)}(k)$ is a compact operator: For this, let $(\varphi_j)_{j \in \nz} \subset L^2(\partial \rz^2_+)$ be a bounded sequence with bound $M > 0$, i.e., $\|\varphi_j\|_{L^2(\partial \rz^2_+)}<M$ for all $j \in \nz$. Due to \eqref{EquationXXX} we observe that $B^{(n)}(k)\varphi_j \in H^1(\partial\rz_+^2)$ and due to the compact embedding of $H^1(I)$ into $L^2(I)$, for any bounded interval $I$, one is able to find a convergent subsequence by restricting $B^{(n)}(k)\varphi_j$ to a (fixed) interval $I_{i,m}$. Furthermore, employing the Bernstein-Cantor diagonal argument one finally obtains a subsequence, again denoted by $(B^{(n)}(k)\varphi_j)_{j\in \nz}$, that converges on any interval $I_{i,m}$. 

Since $|\mathfrak{G}^{(n)}(k)|\leq|\mathfrak{G}^{(0)}(k)|$, Lemma~\ref{aux1} is valid for $B^{(n)}(k)$ as well and we arrive at 
\be
\label{compact1}
\ba
  \|B^{(n)}(k)\varphi_{k}-B^{(n)}(k)\varphi_{l}   \|^2_{L^2(\partial\rz^2_+)}&=\sum\limits_{i=1,2}\sum\limits_{m=1}^{\infty}  \|B^{(n)}(k)\varphi_{k}-B^{(n)}(k)\varphi_{l}   \|^2_{L^2(I_{i,m})}\\
&\leq\epsilon_1+\sum\limits_{i=1,2}\sum\limits_{m=M}^{\infty}  \|B^{(n)}(k)(\varphi_{k}-\varphi_{l})   \|^2_{L^2(I_{i,m})}\\
&\leq \epsilon_1+\epsilon_2 \ , 
\ea
\ee
for $k,l$ and $M$ large enough. We hence conclude that $B^{(n)}(k)\varphi_{j}$ converges in $L^2(\partial \rz^2_+)$ which proves compactness of $B^{(n)}(k)$.

Finally, using \eqref{aux5} in the proof of Lemma \ref{aux1} we can deduce, after a suitable application of Young's and H\"older's inequality, that $\|B(k)-B^{(n)}(k)\|_{L^2(\partial \rz^2_+)      \rightarrow  L^2(\partial \rz^2_+)}      \rightarrow  0$. Hence compactness of $B(k)$ follows by \cite[Satz~3.2]{Weidmann:2000}.  
\end{proof}
We now prove an integration by parts formula which will be useful later on. Using a different method as in \cite[Lemma~2.2]{BrascheExnerKuperinSeba92}) we extend the result to $k\in\overline{\kz_+}$.  
\begin{lemma}[Integration by parts formula]
\label{BoundaryIntegralParts}
Let $k \in \overline{\kz_+}$ and $\psi \in H^1(\rz^2_+)$ with bounded support be given. Then 
\begin{equation}
\label{RelationI}
\langle \nabla (\tilde{B}_0(k)^{\ast}\varphi),\nabla \psi \rangle_{L^2(\rz^2_+)}-(k^2)^{\ast}\langle \tilde{B}_0(k)^{\ast}\varphi,\psi \rangle_{L^2(\rz^2_+)}=\langle \varphi,\psi_{\bv} \rangle_{L^2(\partial \rz^2_+)}
\end{equation}
holds for all $\varphi \in L^2(\partial \rz^2_+)$. If $k\in\kz_+$, the bounded support requirement can be dropped.
\end{lemma}
\begin{proof}
Pick $k\in\overline{\kz_{+}}$, $\psi\in H^1(\rz^2_+)$ and $\varphi\in L^2(\partial\rz^2_+)$ both with bounded support. Then a standard integration by parts yields
\begin{equation}
\begin{split}
\langle \nabla (\tilde{B}_0(k)^{\ast}\varphi),\nabla \psi \rangle_{L^2(\rz^2_+)}&= \int_{\rz^2_+}\ud \bs{y}\int_{\partial \rz^2_+}\ud \textbf{x}\ \overline{\varphi(\textbf{x})\nabla \mf{G}^{(0)}(k)(\textbf{x},\bs{y})}\nabla \psi(\bs{y})\\
&=\int_{\rz^2_+}\ud \bs{y}\int_{\partial \rz^2_+}\ud \textbf{x}\ \overline{\varphi(\textbf{x})}  [(-\Delta-(k^2)^{\ast}) \mf{G}^{(0)}(k)^{\ast}(\textbf{x},\bs{y})   ] \psi(\bs{y})\\ &+(k^2)^{\ast}\int_{\rz^2_+}\ud \bs{y}\int_{\partial \rz^2_+}\ud \textbf{x}\ \overline{\varphi(\textbf{x})\mf{G}^{(0)}(k)(\textbf{x},\bs{y})}\psi(\bs{y})\ .
\end{split}
\end{equation}
Employing relation \eqref{ResolventKernelFree} we get, extending $\psi$ to $\rz^2$ by $\mc{R}\psi$,
\begin{equation}
\begin{split}
\int_{\rz^2_+}\ud \bs{y}\int_{\partial \rz^2_+}&\ud \textbf{x}\ \overline{\varphi(\textbf{x})}  [(-\Delta-(k^2)^{\ast}) \mf{G}^{(0)}(k)^{\ast}(\textbf{x},\bs{y})   ] \psi(\bs{y})=\\
&\int_{\rz^2}\ud \bs{y}\int_{\partial \rz^2_+}\ud \textbf{x}\ \overline{\varphi(\textbf{x})}  [(-\Delta-(k^2)^{\ast}) \mf{G}(k)^{\ast}(\textbf{x},\bs{y})   ] \psi(\bs{y})\ .
\end{split}
\end{equation}
Taking the relation $(-\Delta-(k^2)^{\ast}) \mf{G}(k)^{\ast}(\textbf{x},\bs{y})=\delta(\textbf{x}-\bs{y})$ into account then yields the statement for $\varphi$ and $\psi$ with bounded support.

If $\varphi \in L^2(\partial \rz^2_+)$ has no bounded support, one picks a sequence $(\varphi_n)_{n \in \nz} \subset L^2(\partial \rz^2_+)$ of functions of bounded support, converging to $\varphi$. Relation~\eqref{RelationI} then follows by Lemma~\ref{cont}.

Finally, if $k \in \kz_+$ it is obvious by the previous steps that $\psi \in H^1(\rz^2_+)$ doesn't need to have bounded support.
\end{proof}
\begin{lemma}
\label{BSigmaInjective}
For $k\in\overline{\kz_+}$, the operator $\tilde{B}_0(k)^{\ast}:L^2(\partial \rz^2_+)     \rightarrow  H^1_{\loc}(\overline{\rz^2_+})$ is injective.
\end{lemma}
\begin{proof}
Let $\varphi\neq 0 \in L^2(\partial \rz^2_+)$ be such that $\tilde{B}_0(k)^{\ast}\varphi=0$. Then, by Lemma~\ref{BoundaryIntegralParts} we conclude that $\langle \varphi,\psi_{bv}\rangle_{L^2(\partial \rz^2_+)}=0 $ for all $\psi \in H^1(\rz^2_+)$ with bounded support. Since the boundary values of all $H^1$-functions with bounded support form a dense subset of $L^2(\partial\rz^2_+)$ we conclude that $\varphi\equiv 0$, being a contradiction.
\end{proof}
In the next result we investigate the kernel of the operator $\eins-B(k)$.
\begin{lemma}
\label{PropositionInvertible}
Let the assumption of Lemma \ref{compactop} be satisfied and $\varphi$ be in the kernel of $1-B(k):L^2(\partial \rz^2_+)      \rightarrow  L^2(\partial \rz^2_+)$. If $k\in\kz_+$, then  
\be
\label{p1}
s(\tilde{B}_0(k)^{\ast}\sgn{\sigma}\sqrt{|\sigma|}\varphi,\psi)-\langle k^2\tilde{B}_{0}(k)^{\ast}\sgn{\sigma}\sqrt{|\sigma|}\varphi,\psi \rangle_{L^2(\rz^2_+)}=0
\ee
for all $\psi\in H^1(\rz^2_+)$. If $k\in\overline{\kz_+}$, then \eqref{p1} holds for all $\psi\in H^1(\rz^2_+)$ with bounded support.
\end{lemma}
\begin{proof}
Let $\varphi \in L^2(\partial \rz^2_+)$ such that $\varphi \in \ker (\eins-B(k))$. Due to the eigenvalue equation we may infer 
\be
\sgn{\sigma}\sqrt{|\sigma|}B(k)\varphi=\sigma \tilde{B}(k)\sgn{\sigma}\sqrt{|\sigma|}\varphi=\sgn{\sigma}\sqrt{|\sigma|}\varphi
\ee
and therefore $\tilde{\varphi}:=\sgn{\sigma}\sqrt{|\sigma|}\varphi$ is an element of the kernel for $\eins-\sigma\tilde{B}(k)$. Now we are able to apply Lemma \ref{BoundaryIntegralParts} and calculate 
\begin{equation}
\ba
&s(\tilde{B}_{0}(k)^{\ast}\tilde{\varphi},\psi)-\langle k^2\tilde{B}_{0}(k)\tilde{\varphi},\psi \rangle_{L^2(\rz^2_+)}\\
&\hspace{0.5cm}=\langle\nabla(\tilde{B}_{0}(k)\tilde{\varphi}),\nabla\psi\rangle_{L^2(\rz^2_+)}-\langle\sigma \tilde{B}(k)\tilde{\varphi}, \psi\rangle_{L^2(\partial \rz^2_+)}-\langle k^2\tilde{B}_0(k)\tilde{\tilde{\varphi}},\psi\rangle_{L^2(\rz^2_+)} \\
&\hspace{0.5cm}=\langle\tilde{\varphi},\psi\rangle_{L^2(\partial \rz^2_+)}-\langle \sigma\tilde{B}(k)\tilde{\varphi},\psi\rangle_{L^2(\partial \rz^2_+)} \\
&\hspace{0.5cm}=\langle \tilde{\varphi}-\sigma\tilde{B}(k)\tilde{\varphi}, \psi \rangle_{L^2(\partial \rz^2_+)}=0
\ea
\end{equation}
for $\psi\in H^1(\rz^2_+)$ with bounded support. The claim now follows by an analogous reasoning as in the proof of Lemma \ref{BoundaryIntegralParts}.
\end{proof}
\section{Scattering properties}
\label{Scattering}
In this section we will discuss the scattering properties of our system. First we prove asymptotic completeness of the wave operators. Then we continue the discussion of the scattering properties on a more formal level, constructing the scattering solutions and the (on-shell) scattering amplitude.
\subsection{Existence and completeness of wave operators and an eigenvalue characterizing equation}
We recall that we denote by $E_0$ and $E$ the spectral measure of $-\Delta_{0}$ and $-\Delta_{\sigma}$, respectively. Analogously we denote by $P_0^{(a)}$ and $P^{(a)}$ the projections on the corresponding absolutely continuous subspaces. Moreover, unless stated otherwise, the intervals $I\subset\rz$ are assumed to be closed.
\begin{defn}{\cite[p.~28]{Yafaev:2010}}
\label{asymp}
The wave operators $W_{\pm}(-\Delta_{\sigma},-\Delta_{0};E_0(I))$ and \linebreak $W_{\pm}(-\Delta_{0},-\Delta_{\sigma};E(I))$  are defined by
\be
\label{asymp1}
\ba
&W_{\pm}(-\Delta_{\sigma},-\Delta_{0};E_0(I)):=\st-\lim\limits_{t     \rightarrow \pm\infty}\ue^{-\ui t\Delta_{\sigma}}E_0(I)\ue^{\ui t\Delta_{0}}P_0^{(a)}\ ,\\
&W_{\pm}(-\Delta_{0},-\Delta_{\sigma};E(I)):=\st-\lim\limits_{t     \rightarrow \pm\infty}\ue^{-\ui t\Delta_{0}}E(I)\ue^{\ui t\Delta_{\sigma}}P^{(a)}\ ,
\ea
\ee
provided the strong limits exist.
\end{defn}
\begin{defn}{\cite[pp.~28,29]{Yafaev:2010}}
\label{asmp2}
We say that $W_{\pm}(-\Delta_{\sigma},-\Delta_{0};E_0(I))$ and \linebreak $W_{\pm}(-\Delta_{0},-\Delta_{\sigma};E(I))$ are complete iff
\be
\label{asmp3}
\ba
&\ran W_{\pm}(-\Delta_{\sigma},-\Delta_{0};E_0(I))=\ran P^{(a)}E(I)\ ,\\
&\ran W_{\pm}(-\Delta_{0},-\Delta_{\sigma};E(I))=\ran P^{(a)}_0E_0(I)\ ,.
\ea
\ee
\end{defn}
\begin{remark}
If $I=\rz$, then $E(I)=E_0(I)=\eins$ and we omit $E_0$ and $E$ in $W_{\pm}$. 
\end{remark}
We are now in the position to formulate the first main theorem.
\begin{theorem}[Existence and completeness of wave operators]
\label{TheoremWaveOperators}
Let $\sigma\in L^{\infty}(\rz_+)$ be such that $\sigma(x)=\Or(|x|^{-1-\varepsilon})$, $\varepsilon>0$, $x     \rightarrow \infty$. Then the wave operators $W_{\pm}(-\Delta_{\sigma},-\Delta_{0})$ and $W_{\pm}(-\Delta_{0},-\Delta_{\sigma})$ exist and are complete.
\end{theorem}
\begin{proof}
Lemma~\ref{gg01}, Lemma~\ref{hsmooth1} and Lemma~\ref{compactop} show that the assumption of \cite[Theorem~6.1,~p.~33]{Yafaev:2010} are satisfied with $\theta=\frac{\varepsilon}{2}$ and $m=1$ for every compact $I$. Since $\sigma(-\Delta_0)=\overline{\rz}_+$ we can find a sequence of compact $I_n$ such that $\sigma(-\Delta_0)=\cup_n I_n$. Now \cite[Theorem~6.5,~p.~34]{Yafaev:2010} proves the claim.
\end{proof}
For the next proposition observe that the asymptotics of $k_{\pm}(\lambda,\epsilon):=\sqrt{\lambda\pm\ui\epsilon}$, $\lambda,\epsilon\in\rz_+$, in the limit $\epsilon     \rightarrow 0$ is given by
(in terms of our convention)
\be
\label{k1}
k_{\pm}(\lambda,\epsilon)=
\begin{cases}
\sqrt{\lambda}+\frac{\epsilon}{2\sqrt{\lambda}}+\Or(\epsilon^2)\ ,&  \mbox{for the $+$-case}\ , \\
-\sqrt{\lambda}+\frac{\epsilon}{2\sqrt{\lambda}}+\Or(\epsilon^2)\ ,&  \mbox{for the $-$-case}\ .
\end{cases}
\ee
With \eqref{k1} in mind we adapt a definition of \cite[p.~33]{Yafaev:2010} and define the sets $\mc{N}_\pm\subset\rz$  by, $k\geq 0$,
\be
\label{n1}
\pm k\in\mc{N}_{\pm}\quad   \Leftrightarrow \quad \exists \psi\in L^2(\partial\rz^2_+): \quad -\psi=\lim\limits_{\epsilon     \rightarrow  0^+}B(\pm k+\ui\epsilon)\psi=:B(\pm k+\ui0)\psi\ .
\ee
We obtain
\begin{prop}
\label{n2}
Under the assumptions of Theorem~\ref{TheoremWaveOperators}, the set $\mc{N}:=\mc{N}_+\cup \mc{N_-}$ is closed and has Lebesgue measure zero. Moreover, the operator valued function $k     \rightarrow (\eins+B(k))^{-1}$ exists on $\kz_+$ and is H\"older continuous with exponent $\varepsilon$ up to the cut $\rz$ with the exceptional set $\mc{N}$. Moreover, the spectrum on $\overline{\rz}_+\setminus\mc{N}$ is absolutely continuous.  
\end{prop}
\begin{proof}
The proof uses \cite[Theorem~6.3,~p.~34]{Yafaev:2010} and is analogous to the proof of Theorem~\ref{TheoremWaveOperators}. We only have to take into account that the H\"older continuity in \cite[Theorem~6.3]{Yafaev:2010} is w.r.t. $\lambda$ in the resolvent set. However, $|k_1^2-k^2|=|k_1-k_2||k_1+k_2|$ shows that this is also true for $k=\sqrt{\lambda}$.   
\end{proof}
The next Lemma is important for a further analysis of the set $\mc{N}$. It can be proved analogous to Theorem~\ref{TheoremWaveOperators} using \cite[Proposition~6.7,~p.~35]{Yafaev:2010} and Lemma~\ref{spectralmeasure10}.
\begin{lemma}
\label{n3}
For $k\in\mc{N}$, $k^2={\lambda}$, we have
\be
\label{n4aa}
\frac{\ud}{\ud \lambda}\langle G_0(G_0E_0(I_\lambda))^{\ast}\psi,\psi\rangle_{L^2(\partial\rz^2_+)}=0  
\ee
or equivalently
\be
\label{n4aatr}
\Gamma({\lambda})(G_0E_0(I_\lambda))^{\ast}\psi=0 \ ,
\ee
where $I_{\lambda}=[\lambda-\delta,\lambda]$ for some $\delta>0$.
\end{lemma}
We need some more results which are in the spirit of \cite[Lemma~9.4,~p.~99]{Yafaev:2010}. For this let us recall that, for $\psi,\phi\in L^2(\rz^2_+)$ and $z$ in the resolvent set, we have the identity 
\be
\label{n4}
\langle\psi,R_0(z)\phi\rangle_{L^2(\rz^2_+)}= \int_{\rz_+}(\lambda-z)^{-1}\langle\Gamma_0(\lambda)\psi,\Gamma_0(\lambda)\phi\rangle_{L^2(\sz^1)}\ud \lambda\ .
\ee
Note that this identity can be deduced, for instance, from \cite[(1.4),~p.~4]{Yafaev:2010} and Lemma~\ref{spectralmeasure10}. 

We are now ready to prove the main ingredient in order to construct generalized eigenfunctions by a limiting process for the resolvent, letting $k$ approach the real line from above. We use the method of \cite[Lemma~9.4,~p.~99]{Yafaev:2010}. 
\begin{lemma}
\label{bootsrap1a}                                                                 
Let $\sigma\in L^{\infty}(\rz_+)$ be such that $\sigma(x)=\Or(  x^{-1-\varepsilon})$, $\varepsilon>{0}$, $x     \rightarrow \infty$. Assume that for $k\in\overline{\kz_+}$ we have, $\psi\in L^2(\partial\rz^2_+)$,
\be
\label{bootstrap5a}
\sigma\tilde{B}(k)\sgn{\sigma}\sqrt{|\sigma|}\psi=\sgn{\sigma}\sqrt{|\sigma|}\psi\ .
\ee \label{EquationProofScatteringSolutions}
Then $\sqrt{|\sigma|}\psi\in L^{\infty}(\partial\rz^2_+)$ and
\be
\label{bootstrap2a}
  \|\sqrt{|\sigma|}\psi   \|_{L^{\infty}( (y,\infty)\times (y,\infty))}=C(1+y)^{-1-\varepsilon}  \|\psi   \|_{L^2(\partial\rz_+^2)}
\ee
for some $C>0$ independent of $\psi$ and $y$.
\end{lemma}
\begin{proof}
We consider the case $x=(x,0)$, the other case being analogous. 

Starting from~\eqref{EquationProofScatteringSolutions} and taking into account the asymptotics \eqref{AsymptoticsI} and \eqref{AsymptoticsII}, we obtain applying H\"older's inequality
\be
\ba
&| \int_{\rz_+}\mathfrak{G}(k)(|x_1-y_1|)( \sgn{\sigma}\sqrt{|\sigma|}\psi)(y_1)\ud y_1|\leq C   \|\mathfrak{G}(k)   \|_{L^2(\rz)}  \|\psi   \|_{L^2(\partial\rz_+^2)}\leq\tilde{C}\ ,
\ea
\ee
for some $\tilde{C}>0$. A multiplication with $\sigma(x)=\Or(  x^{-1-\varepsilon})$, $\varepsilon>{0}$, $x      \rightarrow \infty$, then proves the claim.
\end{proof}
\begin{lemma}
\label{bootsrap10aa}                                                                 
Let $\psi\in L^{2}(\partial\rz^2_+)$ satisfy the assumptions in Lemma~\ref{bootsrap1a}. Then, for some $C>0$ independent of $\psi$, we have
\be
\label{bootsrap10aauy}
  \|\Gamma_0(\lambda)(G_iE_0(I))^{\ast}\psi   \|_{\sz^1}\leq {\lambda^{-\frac{1}{12}}}{C  \|\psi   \|_{L^2(\partial\rz_+^2)}},\quad\lambda     \rightarrow \infty\ ,
\ee
uniformly for every compact $I\subset\overline{\rz}_+$.
\end{lemma}
\begin{proof}
We observe that, due to Lemma~\ref{strongcont1}, the term $  \|\Gamma_0(\lambda)(G_iE_0(I))^{\ast}\psi   \|^2_{\sz^1}$ allows for $\lambda$ in the interior of $I$ an expression as a sum with terms of the form \eqref{ap12}. Observing that $k=\sqrt{\lambda}$ then proves the claim.
\end{proof}
\begin{lemma}
\label{bootsrap1}                                                                 
Let $\sigma\in L^{\infty}(\rz_+)$ be such that $\sigma(x)=\Or(  x^{-1-\varepsilon})$, $\varepsilon>{0}$, $x     \rightarrow \infty$. Assume that, $\psi\in L^2(\partial\rz^2_+)$,
\be
\label{bootstrap5}
\Gamma_0(\tilde{\lambda})(G_0E_0(I_\lambda))^{\ast}\psi=0\ ,
\ee
for $\tilde{\lambda}$ being in the interior of $I_\lambda$. Then for $k$ such that $k^2=\tilde{\lambda}$ we have
\be
\label{bootstrap2}
\lim_{\epsilon     \rightarrow  0}{\tilde{B}_0}(k+\ui\epsilon)^{\ast}\sqrt{  |\sigma   |}\psi=:{\tilde{B}_0}(k+\ui0)^{\ast}\sqrt{  |\sigma   |}\psi\in L^{2}(\rz^2_+)\ .
\ee
\end{lemma}
\begin{proof}
Due to the asymptotic behaviour of $\sigma$, \eqref{AsymptoticsI} and \eqref{AsymptoticsII} we can deduce $\tilde{B}_0(k+\ui0)\sqrt{  |\sigma   |}\psi\in L^{\infty}_{\loc}(\rz^2_+)$. To prove the r.h.s. of \eqref{bootstrap2} we are hence going to show that 
\be
\label{bootstrap3}
  |\langle{\tilde{B}_0(k+\ui0)}^{\ast}\sqrt{  |\sigma   |}\psi,\phi\rangle_{L^2(\rz^2_+)}   |\leq C  \|\psi   \|_{L^2(\partial\rz^2_+)}  \|\phi   \|_{L^2(\rz^2_+)}
\ee
for every $\phi\in L^2( \rz^2_+)$ with bounded support. A density argument in combination with the representation theorem of Riesz \cite[Satz~2.16]{Weidmann:2000} then imply ${\tilde{B}_0(k+\ui0)}^{\ast}\sqrt{  |\sigma   |}\psi\in L^2(\rz_+^2)$. 

Let $I_n:=[n+\eta,n+\eta+1]$, $n\in\gz$, with a suitable $\eta$ such that $\tilde{\lambda}$ is an element of the interior of such an interval. Let $I_{\tilde{\lambda}}$ be this interval. We get
\be
\label{bootstrap6}
\ba
&  |\langle\tilde{B}^{\ast}_0(k+\ui0)\sqrt{  |\sigma   |}\psi,\phi\rangle_{L^2(\rz^2_+)}   |\leq\sum\limits_{n=-\infty}^{\infty}  |\langle{R}_0(k+\ui0)(G_0E_0(I_n))^{\ast}\psi,\phi\rangle_{L^2(\rz^2_+)}   |\\
&\leq \int_{I_{\tilde{\lambda}}}\frac{1}{|\lambda-\tilde{\lambda}-\ui0|}  |\langle\Gamma(\lambda)(G_0E_0(I_n))^{\ast}\psi,\Gamma(\lambda)\phi\rangle_{L^2(\sz^1)}   |\ud\lambda\\
&+\sum\limits_{\{I_n\}\setminus I_{\tilde{\lambda}}} \int_{I_n}\frac{1}{|\lambda-\tilde{\lambda}-\ui0|}  |\langle\Gamma(\lambda)(G_0E_0(I_n))^{\ast}\psi,\Gamma(\lambda)\phi\rangle_{L^2(\sz^1)}   |\ud\lambda\ .
\ea
\ee
Regarding the second term on the r.h.s. of \eqref{bootstrap6} we obtain, using Lemma~\ref{bootsrap10aa},
\be
\label{ine1}
\ba
&  |\sum\limits_{\{I_n\}\setminus I_{\tilde{\lambda}}} \int_{I_n}\frac{1}{|\lambda-\tilde{\lambda}-\ui0|}  |\langle\Gamma(\lambda)(G_0E_0(I_n))^{\ast}\psi,\Gamma(\lambda)\phi\rangle_{L^2(\sz^1)}   |\ud\lambda   |\\
&\hspace{0.25cm}\leq \int_{0}^{\infty}  \|\psi   \|_{L^2(\partial\rz_+^2)}\frac{C}{1+|\lambda|^{\frac{13}{12}}}  \|\Gamma(\lambda)\phi   \|_{L^2(\sz^1)}\leq C  \|\psi   \|_{L^2(\partial\rz_+^2)}  \|\phi   \|_{L^2(\rz_+^2)}
\ea
\ee
for a suitable $C>0$. Regarding the first term in \eqref{bootstrap6} we use \eqref{bootstrap5} and Lemma \ref{hsmooth1} to obtain
\be
\label{bootstrap7ac}
\ba
  \|\Gamma(\lambda)(G_0E_0(I_{\tilde{\lambda}}))^{\ast}\psi   \|_{L^2( \sz^1)}&=  \|( \Gamma(\lambda)(G_0E_0(I_{\tilde{\lambda}}))^{\ast}-\Gamma(\tilde{\lambda})(G_0E_0(I_{\tilde{\lambda}}))^{\ast})\psi   \|_{L^2( \sz^1)}\\
&\leq|\lambda-\tilde{\lambda}|^{\frac{\varepsilon}{2}}  \|\psi   \|_{ L^2(\partial\rz_+^2)}
\ea
\ee
which then yields
\be
\label{bootstrap8ac}
\ba
&  | \int_{I_{\tilde{\lambda}}}\frac{1}{|\lambda-\tilde{\lambda}-\ui0|}  |\langle\Gamma(\lambda)(G_0E_0(I_n))^{\ast}\psi,\Gamma(\lambda)\phi\rangle_{L^2(\sz^1)}   |\ud\lambda   |\\
&\hspace{0.5cm}\leq C  \|\psi   \|_{L^2(\partial\rz_+^2)}  \|\phi   \|_{L^2(\rz_+^2)}
\ea
\ee
fora suitable $C>0$. Plugging \eqref{ine1} and \eqref{bootstrap8ac} in \eqref{bootstrap6} proves the claim.
\end{proof}
We are now able to specify the set $\mc{N}$.
\begin{theorem}
\label{th1}
Let $\sigma\in L^{\infty}(\rz_+)$ be such that $\sigma(x)=\Or(|x|^{-1-\varepsilon})$, $\varepsilon>0$, $x     \rightarrow \infty$. Then
\be
\mc{N}=\lgk k\in\kz: \ k^2\in\sigma_{pp}(-\Delta_{\sigma}) \cap [0,\infty)\rgk\ .
\ee
\end{theorem}
\begin{proof}
We only have to show $\subset$ due to Lemma \ref{PropositionInvertible}. Pick $k\in\mc{N}$. By definition there exists a function $\psi\in L^2(\partial\rz^2_+)$ such that $B(\pm k+\ui 0)\psi=\psi$ and by Theorem~\ref{bootsrap1} we conclude that ${\tilde{B}_0}^{\ast}(k+\ui0)\sqrt{  |\sigma   |}\psi\in L^{2}(\rz^2_+)$ and ${\tilde{B}_0}^{\ast}(k+\ui0)\sqrt{  |\sigma   |}\psi \in H^1_{loc}(\rz^2_+)$. Furthermore, by Lemma~\ref{PropositionInvertible} we have
\be
\label{p1a}
\ba
& \int_{\rz^2_+}\overline{\nabla\tilde{B}_0(k)^{\ast}\sgn{\sigma}\sqrt{|\sigma|}\psi} \nabla\varphi\ \ud\bs{x}\\
&\hspace{0.25cm}= \int_{\partial\rz^2}\overline{({\tilde{B}_0(k)}^{\ast}\sgn{\sigma}\sqrt{|\sigma|}\psi)_{\bv}}\varphi_{\bv}\ \ud y+\langle k^2\tilde{B}_{0}(k)^{\ast}\sgn{\sigma}\sqrt{|\sigma|}\psi,\varphi \rangle_{L^2(\rz^2_+)}
\ea
\ee
for all $\varphi\in H^1(\rz^2_+)$ with bounded support. However, the r.h.s. of \eqref{p1a} exists for all $\varphi\in H^1(\rz^2_+)$, depending continuously on $\varphi$. Now, the representation theorem of Riesz \cite[Satz~2.16]{Weidmann:2000} implies that $\nabla\tilde{B}_0(k)^{\ast}\sgn{\sigma}\sqrt{|\sigma|}\psi \in L^2(\rz^2_+)$ and hence ${\tilde{B}_0}^{\ast}(k+\ui0)\sqrt{  |\sigma   |}\psi$ satisfies \eqref{p1} for all $\varphi \in H^1(\rz^2_+)$. This shows that $k^2 \geq 0$ is an eigenvalue of $-\Delta_{\sigma}$.
\end{proof}
\begin{cor}
Let $\sigma\in L^{\infty}(\rz_+)$ have bounded support. Then
\be
\label{qwer}
\mc{N}=\emptyset\ .
\ee
\end{cor}
\begin{proof}
	The statement follows readily by the Theorems~\ref{ee},~\ref{ee1} and \eqref{th1}.
\end{proof}
\subsection{Generalized eigenfunctions and the on-shell scattering amplitude}
In this section we will construct the generalized eigenfunctions (or scattering solutions) associated with $-\Delta_{\sigma}$ and subsequently derive an expression for the on-shell scattering amplitude. Furthermore, in the limit of weak coupling, we obtain an approximation of the scattering amplitude which also illustrates the non-separability of the model.

According to the celebrated Lippmann-Schwinger equation \cite[p.~98]{ReeSim79}, for a Sch\"odinger operator $-\Delta+V$ in $\rz^2$, the scattering solutions $\psi^{+}_{\bs{k}}$ are given by, $\bs{k}=k\omega_{\bs{k}}$,
\begin{equation}
\label{scattering1}
\ba
\psi^{+}_{\bs{k}}(\bs{x})&=\ue^{\ui\langle\bs{k},\bs{x}\rangle}-\lim\limits_{\epsilon     \rightarrow 0^+} \int_{\rz^2}\mf{G}(k+\ui\epsilon)(\bs{x},\bs{y})V(\bs{y})\psi^{+}_{\bs{k}}(\bs{x})\ \ud\bs{y}\\
&=\ue^{\ui\langle\bs{k},\bs{x}\rangle}-(  R_0(k+\ui0)V\psi^{+}_{\bs{k}})(\bs{x})\ ,
\ea
\end{equation}
where $R_0$ is the free resolvent of $(-\Delta,H^2(\rz^2))$. Formally, \eqref{scattering1} is equivalent to $\psi^{+}_{\bs{k}}=(\eins+R_0(k+\ui0)V)^{-1}\ue^{\ui\langle\bs{k},\cdot\rangle}$ and plugging this again into \eqref{scattering1} we arrive at
\be
\label{sc1a}
\psi^{+}_{\bs{k}}(\bs{x})=\ue^{\ui\langle\bs{k},\bs{x}\rangle}-(  R_0(k+\ui0)V(\eins+R_0(k+\ui0)V)^{-1}\ue^{\ui\langle\bs{k},\cdot\rangle})(\bs{x})\ .
\ee
This in turn is equivalent to 
\be
\label{LS1}
\ba 
&\psi^{+}_{\bs{k}}(\bs{x})=\ue^{\ui\langle\bs{k},\bs{x}\rangle}\\
\hspace{0.5cm} &-(R_0(k+\ui0)\sgn(V)\sqrt{|V|}(\eins+\sqrt{|V|}R_0(k+\ui0)\sgn(V)\sqrt{|V|})^{-1}\sqrt{|V|}\ue^{\ui\langle\bs{k},\cdot\rangle})(\bs{x})\ .
%
\ea
\ee

To get an idea of how \eqref{LS1} translates into our setting we first observe that scattering solutions in the free case where $V\equiv 0$ are not only plane waves but symmetrised plane waves $S[\psi_{\bs{k}}]$, i.e., 
\be
\label{sk}
S[\psi_{\bs{k}}](x_1,x_2):=\ue^{\ui (k_1x_1+k_2x_2)}+\ue^{-\ui (k_1x_2-k_2x_2)}+\ue^{\ui (k_1x_1-k_2x_2)}+\ue^{-\ui (k_1x_1+\ui k_2x_2)}\ .
\ee
%
%
The reason for this is that the free operator is the Laplacian on $L^2(\rz^2_+)$ subjected to Neumann boundary conditions. To understand this from a physics point of view one observes that a single free particle on the half-line is described by a superposition of an incoming and an outgoing plane wave of same amplitude, due to the perfect reflection at the origin.

Furthermore, since the two-particle potential $V$ is singular and has support on the boundary $\partial \rz^2_+$ only, we conclude that, comparing \eqref{LS1} with \eqref{Rsigma}, that the scattering solution should be of the form, $\bs{k}=\omega k$,
\begin{equation}
\label{ScatteringSolution}
\psi^{+}_{\bs{k}}:=S[\psi_{\bs{k}}]-B_{1}(k+\ui0)^{\ast}( \eins+B(k+\ui0))^{-1}\sqrt{|\sigma|}(  S[\psi_{\bs{k}}])_{\bv}\ .
\end{equation}
We will show that \eqref{ScatteringSolution} is indeed well-defined.

In a first result, we will show that \eqref{ScatteringSolution} is indeed a generalized eigenfunction for $k\in\rz_+$, i.e., $\psi^{+}_{\bs{k}}$ satisfies locally the boundary conditions \eqref{ROBINBC} and fulfills
\begin{equation}
\label{ge1a}
-\Delta\psi^{+}_{\bs{k}}(\bs{x})-k^2\psi^{+}_{\bs{k}}(\bs{x})=0\ , \quad \bs{x} \in \Omega \ ,
\end{equation}
on any open set $\Omega$ which is compactly contained in $\rz^2_+$. For further convenience we use a weak form \eqref{ge1a}.
\begin{defn}
$\psi^{+}_{\bs{k}}\in H^1_{\loc}(\rz^2_+)$ is a generalized eigenfunction iff
\begin{equation}%
\label{EquationProofScattering}
s(\psi^{+}_{\bs{k}},\varphi)=k^2\langle \psi^{+}_{\bs{k}},\varphi\rangle_{L^2(\rz^2_+)}
\end{equation}
holds for all $\varphi \in H^1(\rz_+^2)$ with bounded support.
\end{defn}
\begin{theorem}
\label{scattering3}
Let $k \in \rz_+$ be such that $k^2\notin\sigma_{pp}(-\Delta_{\sigma})$ and $\sigma\in L^{\infty}(\rz_+)$ such that $\sigma(x)=\Or(  x^{-1-\varepsilon})$, $\varepsilon>{0}$, $x     \rightarrow \infty$. Then $\psi^{+}_{\bs{k}}$ as in \eqref{ScatteringSolution} is well-defined and $\psi^{+}_{\bs{k}}$ is a generalized eigenfunction to $-\Delta_{\sigma}$. 
\end{theorem}
\begin{proof} 
We first observe that $\sqrt{|\sigma|}(  S[\psi_{\bs{k}}])_{\bv}\in L^2(\rz^2_+)$. By Theorem \ref{th1} and Proposition \ref{n2} we may conclude that $( \eins+B(k+\ui0))^{-1}\sqrt{|\sigma|}(  S[\psi_{\bs{k}}])_{\bv}\in L^2(\partial\rz^2_+)$. Finally, by Lemma \ref{cont} we conclude that $\psi^{+}_{\bs{k}}\in H^1_{\loc}(\overline{\rz^2_+})$ which shows that $\psi^{+}_{\bs{k}}$ is indeed well-defined. 

Since $S[\psi_{\bs{k}}]$ fulfills Neumann boundary conditions along $\partial \rz^2_+$ one can employ an integration by parts to obtain the relation, $\varphi \in H^1(\rz^2_+)$ with bounded support, 
\begin{equation}
s(S[\psi_{\bs{k}}],\varphi)=k^2\langle S[\psi_{\bs{k}}],\varphi \rangle_{L^2(\rz^2_+)}-\langle \sigma S[\psi_{\bs{k}}],\varphi\rangle_{L^2(\partial \rz^2_+)}\ .
\end{equation}
Now, employing Lemma~\ref{BoundaryIntegralParts} and recalling that $B_{1}^{\ast}=-{\tilde{B}_{0}}^{\ast}\sgn(\sigma)\sqrt{|\sigma|}$ while setting 
\be
\eta:=( \eins+B(k+\ui0))^{-1}\sqrt{|\sigma|}(  S[\psi_{\bs{k}}])_{\bv}
\ee
we obtain
\begin{equation}
\begin{split}
s(B_{1}^{\ast}(k+\ui0)\eta,\varphi)=&k^2\langle B_{1}^{\ast}(k+\ui0)\eta,\varphi \rangle_{L^2(\rz^2_+)}+\langle\eta,\varphi\rangle_{L^2(\partial \rz^2_+)} \\ 
&-\langle (\sigma B_{1}^{\ast}(k+\ui0)\eta)_{\bv},\varphi\rangle_{L^2(\partial \rz^2_+)}\ .
\end{split}
\end{equation}
Since 
\begin{equation}
\ba
&\langle\sgn(\sigma)\sqrt{|\sigma|}\eta,\varphi\rangle_{L^2(\partial \rz^2_+)}+ \langle-\sigma B_{1}^{\ast}(k + \ui0)\eta,\varphi\rangle_{L^2(\partial \rz^2_+)}\\
&\hspace{0.25cm}=\langle-\sgn(\sigma)\sqrt{|\sigma|}\eta,\varphi\rangle_{L^2(\partial \rz^2_+)}+ \langle-\sgn(\sigma)\sqrt{|\sigma|}B(k + \ui0)\eta,\varphi\rangle_{L^2(\partial \rz^2_+)}\\
&\hspace{0.25cm}=-\langle\sigma S[\psi_{\bs{k}}],\varphi\rangle_{L^2(\partial \rz^2_+)}\ ,
\ea
\end{equation}
we arrive at 
\begin{equation}
\label{limit}
s(\psi^{+}_{\bs{k}},\varphi)=k^2\langle \psi^{+}_{\bs{k}},\varphi\rangle_{L^2(\rz^2_+)}\ ,
\end{equation}
thus proving the statement.
\end{proof}
We now want to derive an expression for the the so-called on-shell scattering amplitude. As customary in physics one expects the scattered solution, i.e., the generalized eigenfunction~\eqref{ScatteringSolution} to have the asymptotic form 
\begin{equation}
\psi^{+}_{\bs{k}}(\bs{x}) \sim S[\psi_{\bs{k}}](\bs{x})+f(k,\omega,\omega^{\prime})\frac{\ue^{\ui k\|\bs{x}\|}}{\sqrt{\|\bs{x}\|}}\ ,
\end{equation} 
as $\|\bs{x}\|     \rightarrow  \infty$ and where $f(k,\omega,\omega^{\prime})$ is the scattering amplitude. Here $\omega,\omega^{\prime} \in S^1$, $\bs{k}=k\omega$ and $\bs{x}=\|\bs{x}\|\omega'$. In other words, one formally defines see e.g. \cite{bRASCHE:1992}
\begin{defn}
The formal scattering amplitude is defined by
\begin{equation}
\label{ScatteredSolution}
f(k,\omega,\omega^{\prime}):=\lim\limits_{\substack{\|\bs{x}\|      \rightarrow  \infty\\ \omega^{\prime}=\frac{\bs{x}}{\|\bs{x}\|}}} \sqrt{\|\bs{x}\|}\ue^{-\ui k\|\bs{x}\|}(\psi^{+}_{\bs{k}}-S[\psi_{\bs{k}}])(\bs{x})\ .
\end{equation}
\end{defn}
We will show that \eqref{ScatteredSolution} is well-defined for weak potentials. 
\begin{defn}
\label{scaled}
For $\sigma\in L^{\infty}(\rz_+)$ such that $\sigma(x)=\Or(x^{-1-\varepsilon})$, $x     \rightarrow \infty$, for some $\varepsilon>0$, the scaled potential $\sigma_{\alpha}$ with coupling parameter $\alpha\in\rz_+$ is defined via
\be
\sigma_{\alpha}(x):=\alpha \sigma(x), \quad \alpha\in\rz_+\ .
\ee
\end{defn}
By Definition~\ref{defq}, the potential $\sigma_{\alpha}$ yields a one-parameter family of operators $-\Delta_{\sigma_{\alpha}}$ .
\begin{theorem}
\label{propsca}
Let $\sigma_{\alpha}$ be as in Definition~\ref{scaled} and $k\in\rz_+$ such that $k^2\notin\sigma_{pp}(-\Delta_{\sigma})$. Then the scattering amplitude in \eqref{ScatteredSolution} is well-defined and it is given by
\be
\label{scaf}
\ba
&f(k,\omega,\omega^{\prime})=-2\sqrt{\frac{\ui}{k}}  [\mathcal{F}_{0,1}(\mathcal{R}_{\bv}( \sgn (\sigma_{\alpha})\sqrt{|\sigma_{\alpha}|}h_{\bf k})|_{C_{x_1}})(k_1')\\
&\hspace{0.75cm}+\mathcal{F}_{0,1}(\mathcal{R}_{\bv}( \sgn (\sigma_{\alpha})\sqrt{|\sigma_{\alpha}|}h_{\bf k})|_{C_{x_2}})(k_2')   ]\ ,
\ea
\ee
where $h_{\bs{k}}:=( \eins+B(k+\ui0))^{-1}\sqrt{|\sigma_{\alpha}|}(  S[\psi_{\bs{k}}])_{\bv}$.
\end{theorem}
\begin{proof} 
By \eqref{ScatteringSolution} we may write, $\omega'\neq(1,0),(0,1)$,
\begin{equation}\begin{split}
f(k,\omega,\omega^{\prime})&=\lim_{\substack{\|\bs{x}\|      \rightarrow  \infty\\ \omega^{\prime}=\frac{\bs{x}}{\|\bs{x}\|}}} \sqrt{\|\bs{x}\|}\ue^{-\ui k\|\bs{x}\|}B_{1}(k+\ui0)^{\ast}( \eins+B(k+\ui0))^{-1}\sqrt{|\sigma|}(  S[\psi_{\bs{k}}])_{\bv} \\
&=\lim_{\substack{\|\bs{x}\|      \rightarrow  \infty\\ \omega^{\prime}=\frac{\bs{x}}{\|\bs{x}\|}}} \sqrt{\|\bs{x}\|}\ue^{-\ui k\|\bs{x}\|}B_{1}(k+\ui0)^{\ast}h_{\bs{k}}\ .
\end{split}
\end{equation}
Using the asymptotics for $\sigma_{\alpha}$, \eqref{AsymptoticsI} and Proposition~\ref{n2} in combination with Theorem~\ref{th1}, we may conclude that $h_{\bs{k}}\in L^2(\partial \rz_+^2)$. Also, the pointwise asymptotics \cite[p.~328]{Schwabl:2007}
\begin{equation}
\frac{\ue^{\ui k\|\textbf{x}-\textbf{y}\|}}{\sqrt{\|\textbf{x}-\textbf{y}\|}} \sim \frac{\ue^{\ui k\|\textbf{x}\|}}{\sqrt{\|\textbf{x}\|}}\ue^{-\ui \langle\bs{k'},\textbf{y}\rangle},\quad \|\bs{x}\|      \rightarrow  \infty,\ \bs{k'}:=\omega^{\prime}k\ ,
\end{equation}
holds. Combining this with Lemma~\ref{integralkernel4} and \eqref{AsymptoticsI} we may deduce that
\be
\label{dcon1}
\lim\limits_{\substack{\|\bs{x}\|      \rightarrow  \infty\\ \omega^{\prime}=\frac{\bs{x}}{\|\bs{x}\|}}}\sqrt{\|\bs{x}\|}\ue^{-\ui k\|\bs{x}\|}B_{1}(k+\ui0)^{\ast}(\bs{x},y)=-\sqrt{\frac{1}{-8\pi i k}}\sgn{\sigma_{\alpha}}(y)\sqrt{|\sigma_{\alpha}(y)|}\sum_{j=1}^{4}\ue^{-\ui\langle\bs{k'},\textbf{y}_j\rangle}\ , 
\ee
with $\textbf{y}_j\in\partial{\rz^2_+}$ according to \eqref{ResolventKernelFree}. Moreover, \eqref{AsymptoticsI} and \eqref{AsymptoticsII} together with the supposed properties of $\sigma_{\alpha}$ reveals that for $\omega'\neq(1,0),(0,1)$ there exists a constant $C$ such that, $\omega^{\prime}=\frac{\bs{x}}{\|\bs{x}\|}$, $\|\bs{x}\|$ large, 
\be
\label{dcon2}
\sqrt{\|\bs{x}\|}|B_{1}(k+\ui0)^{\ast}(\bs{x},{y})|\leq \frac{C}{(1+\|{y}\|)^{\frac{1+\varepsilon}{2}}}\quad \mbox{for all} \  {y}\in\partial\rz^2_+\ .
\ee 
Using \eqref{dcon1} and \eqref{dcon2} we may apply Lebesgue dominated converges theorem \cite[15.6~Theorem]{Bauer:2001} for $L^2(\partial \rz^2_+)$ which proves the claim after some straightforward calculations.
\end{proof}
We are now in the position to present an explicit formula for the scattering amplitude in the regime of weak coupling, i.e., $\alpha     \rightarrow 0$.
\begin{prop}
Under the assumptions of Theorem~\ref{propsca}, for $\alpha     \rightarrow  0$, the scattering amplitude possesses a complete asymptotic expansion in powers of $\alpha$ with leading coefficient
\be \label{StreuamplitudeWeakCoupling}
\ba
\ba
&f(k,\omega,\omega^{\prime})=\alpha\sum_{j=1}^2(\hat{\sigma}_k(k^{\prime}_j+k_j)+\hat{\sigma}_k(k^{\prime}_j-k_j)+\hat{\sigma}_k(k_j-k^{\prime}_j)+\hat{\sigma}_k(-k_j-k^{\prime}_j))\\
&\qquad \qquad \qquad \qquad +\Or( \alpha^2) \ ,
\ea\\ ,
\ea
\end{equation}
with
\begin{equation}
\hat{\sigma}_k(\xi)=-4\sqrt{\frac{\ui}{k}}\mc{F}_{0,1}(\mc{R}_{\bv}\sigma)(\xi)\ .
\ee
\end{prop}
\begin{proof}
We rewrite \eqref{scaf} and obtain
\begin{equation}\label{EquationProofScatteringAmplitudeWeak}\begin{split}
f(k,\omega,\omega^{\prime})=&-2\sqrt{\frac{\ui}{2\pi k}} \int_{\rz}\mc{R}_{\bv}( \sgn(\sigma_{\alpha}) \sqrt{|\sigma_{\alpha}|}h_{\bf k})(y,0)\ue^{-\ui k^{\prime}_1 y}\ \ud y \\
& \quad -2\sqrt{\frac{\ui}{2\pi k}} \int_{\rz}\mc{R}_{\bv}( \sgn(\sigma_{\alpha}) \sqrt{\sigma_{\alpha}|}h_{\bf k}))(0,y)\ue^{-\ui k^{\prime}_2 y}\ \ud y\ .
\end{split}
\end{equation}

Due to the terms $\sigma_{\alpha}(x)$ and $\sigma_{\alpha}(y)$ in Lemma \ref{integralkernel2} and the asymptotic behavior \eqref{AsymptoticsI} and \eqref{AsymptoticsII} we may infer that $B(k+\ui0)$ is a bounded operator in $L^2(\partial\rz^2_+)$. Moreover, we see in Lemma \ref{integralkernel2} that the coupling constant acts in $B(k+\ui0)$ simply as a scalar multiplication operator and hence for small $\alpha$ the operator $(\eins+B(k+\ui0))^{-1}$ allows a Neumann series representation
\be
\label{nsc}
(\eins+B(k+\ui0))^{-1}=\eins+\sum\limits_{n=1}^{\infty}\alpha^nA_n(k+\ui0)
\ee
with some (uniformly) bounded operators $A_n(k+\ui0)$. To first order in $\alpha$ we therefore obtain
\begin{equation}
\label{hfw}
\begin{split}
h_{\bs{k}}(y,0)=\sqrt{|\sigma_{\alpha}|} S[\psi_{\bs{k}}] |_{(y,0)}=2\sqrt{|\sigma_{\alpha}(y)|}(\ue^{-ik_1y}+\ue^{+ik_1y})\ , \\
h_{\bs{k}}(0,y)=\sqrt{|\sigma_{\alpha}|} S[\psi_{\bs{k}}] |_{(0,y)}=2\sqrt{|\sigma_{\alpha}(y)|}(\ue^{-ik_2y}+\ue^{+ik_2y})\ .
\end{split}
\end{equation}
Plugging \eqref{hfw} into \eqref{EquationProofScatteringAmplitudeWeak} then yields the statement.
\end{proof}
\begin{remark} Equation~\eqref{StreuamplitudeWeakCoupling} illustrates the non-separability of the singular two-particle interactions, i.e., momentum is exchanged componentwise.  
	\end{remark}
We present an easy example.
\begin{example} We consider the case where $\sigma\in L^{\infty}(\rz_+)$ is a step-potential, i.e.,
\begin{equation}
\sigma(y)=
\begin{cases}
{\sigma}_0, &\text{for} \quad y\in (0,L]\ , \\
0,& \text{else}\ ,
\end{cases}
\end{equation}
where ${\sigma}_0\in\rz$ is some constant. We obtain
\begin{equation}
\ba
\hat{\sigma}_k(\xi)&=-\frac{2\ui {\sigma}_0}{\xi}\sqrt{\frac{\ui}{2\pi k}}(\ue^{-\ui \xi L}-\ue^{\ui \xi L})\\
               &=-\frac{4{\sigma}_0}{\xi}\sqrt{\frac{\ui}{2\pi k}}\sin(\xi L)\ .
\ea
\end{equation}
Hence, if we assume $\xi L << 1$, then $\hat{\sigma}_k(\xi)\approx -{\sigma_0}L\sqrt{\frac{\ui}{2\pi k}}$ (low-energy limit) and we obtain in the weak-coupling limit 
\begin{equation}
|f(k,\omega,\omega^{\prime})|^2 \approx \frac{128}{\pi |k|}( \alpha\sigma_0L)^2\ .
\end{equation}
\end{example}
\subsection*{Acknowledgment}{J.~Kerner would like to thank T.~M\"{u}hlenbruch for helpful discussions and S.~Egger expresses his gratitude towards V.~Lotoreichik for useful comments and for pointing out stimulating references.
}
\appendix
\section{Notation}\label{Notation}

In this paper we use the following notation: the values of $z^{\beta}$, $z\in\kz$, $\beta\in\rz$ are determined by requiring that the branch cut is at $\rz_+$ and $\arg z=\pi$ for $-z\in\rz_+$. Moreover, we put
\begin{itemize} 
	\item $z^{\beta}:=\lim\limits_{\epsilon     \rightarrow  0^+}(z+\ui\epsilon)^{\beta}$\ , $z\in\rz_+$\ , $\beta\in\rz$\ ,
	\item $k:=\sqrt{z}$\ , $z\in\kz\setminus\rz_+$\ ,
\end{itemize}
and we use 
\be
\label{kz+}
\kz_+:=\lgk z\in\kz: \quad \im z>0\rgk\ .
\ee
Note that for $z\in\kz$ we always have $\sqrt{z}\in\overline{\kz_+}$.

Finally, we consider $\partial\rz^2_+$ to be canonically embedded into $\rz^2$.  Unless stated otherwise, bold coordinates refer to an element of $\rz^{2}_+$ and non-bold coordinates to be an element of $\partial\rz^2_+$, i.e., $x=(x_1,0)$ or $x=(0,x_2)$, respectively. 

\section{Some integral estimates}
\label{integralapp}
\begin{lemma}
\label{app1}
Let $\sigma\in L^{\infty}(\rz_+)$ be given with  
\be
\label{app10}
\sigma(x)=\Or(x^{-\alpha}), \quad x     \rightarrow \infty\ ,\ \alpha>0\ .
\ee
Moreover, assume that $g:\rz_+     \rightarrow \kz$ is continuous and satisfies
\be
\label{app7}
g(kx)=
\begin{cases}
\Or(|\ln k x|)\ , & x     \rightarrow  0\ ,\\
\Or(\ue^{-kx})\ , & x     \rightarrow  \infty\ , 
\end{cases}
\ee
with some $k\in\rz_+$. Then, for $\psi\in L^2(\rz_+)$, we have
\be
\label{ap2}
 \int_{\rz_+}g(k|x-y|)\sigma(y)\psi(y)\ud y = \Or(x^{-\alpha})\ .
\ee
\end{lemma}
\begin{proof}
We denote $I_x:=  [x/2,3x/2   ]$ and write
\be
\label{app4}
 \int_{\rz_+}g(k|x-y|)\sigma(y)\psi(y)\ud y= \int_{I_x}g(k|x-y|)\sigma(y)\psi(y)\ud y+ \int_{\rz_+\setminus I_x}g(k|x-y|)\sigma(y)\psi(y)\ud y\ .
\ee
We get, using H\"older's inequality and \eqref{app7},
\be
\label{app5}
\ba
& \int_{0}^{x/2}g(k|x-y|)\sigma(y)\psi(y)\ud y= \int_{x/2}^{x}g(ky)\sigma(  y-\frac{x}{2})\psi(  y-\frac{x}{2})\ud y\\
&\leq\ue^{-kx/2}\cdot C'  \|\psi   \|_{L^2(\rz_+)}
\ea
\ee
for all $x > x_0$ with some $x_0$ and $C'>0$ where $C'$ depends on $x_0$ only. An analogous result holds for the same integral as in \eqref{app5} but with interval $[3x/2,\infty)$.

Moreover, using H\"older's inequality again and \eqref{app7}, \eqref{app10}, we obtain
\be
\label{app6}
 \int_{x/2}^{3x/2}g(k|x-y|)\sigma(y)\psi(y)\ud y\leq\frac{\tilde{C}}{x^{\alpha}}  \|\psi   \|_{L^2(\rz_+)}
\ee
for all $x > x_0$ and some $\tilde{C}>0$ depending on $x_0$ only. Hence, combining \eqref{app5} and \eqref{app6} proves the claim.
\end{proof}
The next lemma provides an sufficient good asymptotic estimates of the l.h.s. of \eqref{bootsrap10aauy}. 
\begin{lemma}
\label{sum}
Let $\psi,\tilde{\psi}\in L^{\infty}(\rz)\cap L^2(\rz)$ be such that
\be
\label{2infty}
\|\cdot\|_{L^{\infty}( (\pm|y|,\pm \infty))}\leq C(1+|y|)^{-1-\varepsilon}\|\cdot\|_{L^2(\rz)}
\ee
holds for some $C > 0$. Then, as $k     \rightarrow \infty$,
\be
\label{ap12}
| \int_{0}^{2\pi}\overline{[ F_{0,1}\psi](g_j(\varphi)k)}{[ F_{0,1}\tilde{\psi}](g_l(\varphi)k)}\ \ud \varphi|\leq{k^{-\frac{1}{3}}}{\tilde{C} \|\psi\|_{L^2(\rz)} \|\tilde{\psi}\|_{L^2(\rz)}}\ ,
\ee
where $g_1=\sin(x)$ and $g_2=\cos(x)$. Moreover, $\tilde{C}$ is independent of $\psi,\tilde{\psi}$.
\end{lemma}
\begin{proof}
We first consider the case $j=2$ and $l=1$: Then the absolute value of the integral in \eqref{ap12} is given by
\be
\label{ap13a}
| \int_{0}^{2\pi} \int_{\rz^2}\overline{\psi}(  y)\tilde{\psi}(  y')\ue^{\ui k(  y\cos(\varphi)-y'\sin(\varphi))}\ud y\ud y'\ud \varphi|\ .
\ee
We split the integral w.r.t. $\varphi$ in two subintervals $[0,\pi]$ and $[\pi,2\pi]$, considering only the first case, the other being similar. We write 
\be
\label{aux}
\ue^{\ui k(  y\cos(\varphi)-y'\sin(\varphi))}=\ue^{\ui k(  y\cos(\varphi))}( \cos(  ky'\sin(\varphi))-\ui\sin(  ky'\sin(\varphi)))
\ee
and set, using \cite[pp.~81]{Magnus:1966},
\be
\label{estimate12}
\ba
&f(y,y',k):= \int_{0}^{\pi}\ue^{\ui k(  y\cos(\varphi)-y'\sin(\varphi))}\ud\varphi=\pi J_{0}(  k(  y^2+y'^2)^{\frac{1}{2}})\\
&\hspace{0.25cm}-\ui\pi Y_{0}(  k(  y^2+y'^2)^{\frac{1}{2}})-\ui2\pi \int_{0}^{\infty}\sin(  yk \cosh(t))\ue^{-y'k\sinh(t)}\ \ud t\ .
\ea
\ee
For $y'\geq{{k}^{-\alpha}}$, $\alpha>0$, we have
\be
\label{estimate1a}
| \int_{0}^{\infty}\sin(  yk \cosh(t))\ue^{-y'k\sinh(t)}\ud t|\leq \int_{0}^{\infty}\ue^{-y'k t}\ud t=\Or( {{k}^{\alpha-1}})\ .
\ee
Since the l.h.s. of \eqref{estimate12} is smooth and bounded w.r.t. to $y$ and $y'$, we deduce that the second plus the third term on the r.h.s. of \eqref{estimate12} is also bounded for $r:=\sqrt{y^2+{y'}^2}<{{k}^{-\alpha}}$. Consequently, taking into account the asymptotics for large values of $k$ of the first two terms on the r.h.s. of \eqref{estimate12} \cite[p.~139]{Magnus:1966} we obtain
\be
\label{rtea}
\ba
&| \int_{\rz^2}f(y,y',k)\overline{\psi}(y)\tilde{\psi}(y')\ud y\ud y'|\leq | \int_{r<{k}^{-\alpha}}f(y,y',k)\overline{\psi}(y)\tilde{\psi}(y')\ud y\ud y'|\\
&\hspace{0.25cm}+| \int_{r\geq{k}^{-\alpha}}f(y,y',k)\overline{\psi}(y)\tilde{\psi}(y')\ud y\ud y'|\leq {C  \|\psi   \|_{L^2(\rz)}\|\tilde{\psi}\|_{L^2(\rz)}}\max\lgk{k}^{-2\alpha},{k}^{\frac{\alpha-1}{2}}\rgk\ ,
\ea
\ee
where $C>0$ can be chosen independently of $\psi$ and $\tilde{\psi}$. Note also that we used the fact that $\psi,\tilde{\psi}$ satisfy \eqref{2infty} in estimating the second integral.

In a next step we consider the case $j=1$ and $l=1$: Then the absolute value of the integral in \eqref{ap12} is given by
\be
\label{ap13}
| \int_{0}^{2\pi} \int_{\rz^2}\overline{\psi}(  y)\tilde{\psi}(  y')\ue^{\ui k(y-y')\cos(\varphi)}\ud y\ud y'\ud \varphi|\ .
\ee
Using \eqref{aux} and \cite[p.~79]{Magnus:1966} we see that
\be
h(y,y',k):=\int_0^{2\pi}\ue^{\ui k(y-y')\cos(\varphi)}\ud \varphi=2\pi J_{0}(  k(  y-y'))\ .
\ee
We make the disjoint decomposition $\rz^2=Q\cup\lgk\rz^2\setminus Q\rgk$ where
\be
\label{Q}
Q=\lgk(x,y)\in\rz^2; \quad |x-y|\leq{{k}^{-\alpha}}\rgk\ .
\ee
We observe that 
\be
\label{estimate1aa}
 \int_{Q}\frac{1}{(  1+|y|)^{\frac{1+\varepsilon}{2}}}\frac{1}{(  1+|y'|)^{\frac{1+\varepsilon}{2}}}\ud y\ud y'\leq C\int_{0}^{{k}^{-\alpha}} \int_{0}^{\infty}\frac{1}{(  1+|y|)^{\frac{1+\varepsilon}{2}}}\ud y\ud y^\prime=\Or( {k}^{-\alpha})\ .
\ee
Using that $\psi,\tilde{\psi}\in L^{\infty}(\rz)$, $\psi,\tilde{\psi}$ satisfy \eqref{2infty}, \eqref{estimate1aa} and the asymptotics of the $J_0$ Bessel function \cite[p.~138]{Magnus:1966} we get 
\be
\label{ineq1as}
\ba
&| \int_{\rz^2}h(y,y',k)\overline{\psi}(y)\tilde{\psi}(y')\ud y\ud y'|\leq C| \int_{Q}\overline{\psi}(y)\tilde{\psi}(y')\ud y\ud y'|\\
&\hspace{0.5cm}+C| \int_{\rz^2\setminus Q}{k}^{\frac{\alpha-1}{2}}\overline{\psi}(y)\tilde{\psi}(y')\ud y\ud y'|\\
&\hspace{0.5cm} <{C\|\psi\|_{L^2(\rz)}\|\tilde{\psi}\|_{L^2(\rz)}}{\max\lgk k^{-\alpha},{k}^{\frac{\alpha-1}{2}}\rgk}\ ,
\ea
\ee
where $C>0$ can be chosen independently of $\psi$ and $\tilde{\psi}$. Combining \eqref{rtea} and \eqref{ineq1as} we may choose $\alpha=\frac{1}{3}$. Finally, the other cases are analogous to the two previous ones.
\end{proof}
%
{
\small
\bibliographystyle{plain}
\bibliography{Literature}
}
\end{document}